\documentclass[aps,pra,twocolumn,superscriptaddress]{revtex4-2} 
\usepackage{amssymb,amsmath}
\usepackage{bm,bbm}
\usepackage{graphicx}
\usepackage[bookmarks=false]{hyperref}
\hypersetup{colorlinks=true,citecolor=blue,linkcolor=red,%
urlcolor=blue,pdfstartview=FitH,bookmarksopen=true}
\usepackage[T1]{fontenc}
\usepackage[osf,sc]{mathpazo}
\usepackage{dsfont}
\usepackage{mathtools}
\usepackage{enumitem}
\usepackage{tikz}
\usetikzlibrary{positioning}
\usepackage{extarrows}

\newcommand{\ket}[1]{\mbox{$ | #1 \rangle $}}
\newcommand{\bra}[1]{\mbox{$ \langle #1 | $}}

\newcommand{\tr}{\mathrm{tr}}

\newcommand{\cP}{{\cal P}}

\newcommand{\cB}{\mathcal{B}}

\newcommand{\cS}{\mathcal{S}}

\newcommand{\cV}{\mathcal{V}}

\newcommand{\I}{\mathrm{i}}
\newcommand{\D}{\mathrm{d}}

\definecolor{light-gray}{gray}{0.95}

\usepackage{amsthm}
\newtheoremstyle{note}      
  {\topsep/2}              	
  {\topsep/2}            	
  {}                        
  {\parindent}             	
  {\itshape}                
  {.---}                    
  {0pt}                     
  {\thmname{#1}\thmnumber{ \itshape#2}\thmnote{ (#3)}} 

\newtheorem{theorem}{Theorem}

\newtheorem{corollary}{Corollary}

\theoremstyle{definition}

\theoremstyle{remark}

\begin{document}
\title{Efficient verification of arbitrary entangled states with homogeneous local measurements}

\author{Ye-Chao Liu}
\affiliation{Key Laboratory of Advanced Optoelectronic Quantum Architecture and Measurement of Ministry of Education, School of Physics, Beijing Institute of Technology, Beijing 100081, China}
\affiliation{Naturwissenschaftlich-Technische Fakult{\"a}t, Universit{\"a}t Siegen, Walter-Flex-Stra{\ss}e 3, 57068 Siegen, Germany}

\author{Yinfei Li}
\affiliation{Key Laboratory of Advanced Optoelectronic Quantum Architecture and Measurement of Ministry of Education, School of Physics, Beijing Institute of Technology, Beijing 100081, China}

\author{Jiangwei Shang}
\email{jiangwei.shang@bit.edu.cn}
\affiliation{Key Laboratory of Advanced Optoelectronic Quantum Architecture and Measurement of Ministry of Education, School of Physics, Beijing Institute of Technology, Beijing 100081, China}
\affiliation{State Key Laboratory of Surface Physics and Department of Physics, Fudan University, Shanghai 200433, China}

\author{Xiangdong Zhang}
\email{zhangxd@bit.edu.cn}
\affiliation{Key Laboratory of Advanced Optoelectronic Quantum Architecture and Measurement of Ministry of Education, School of Physics, Beijing Institute of Technology, Beijing 100081, China}

\date{\today}
%

\begin{abstract}
Quantum state verification (QSV) is the task of relying on local measurements only to verify that a given quantum device does produce the desired target state.
Up to now, certain types of entangled states can be verified efficiently or even optimally by QSV.
However, given an arbitrary entangled state, how to design its verification protocol remains an open problem.
In this work, we present a systematic strategy to tackle this problem by considering the locality of what we initiate as the choice-independent measurement protocols, whose operators can be directly achieved when they are homogeneous. 
Taking several typical entangled states as examples, we show the explicit procedures of the protocol design using standard Pauli projections, demonstrating the superiority of our method for attaining better QSV strategies.
Moreover, our framework can be naturally extended to other tasks such as the construction of entanglement witness, and even parameter estimation.
\end{abstract}

\maketitle

\section{Introduction}
With the rapid development of quantum techniques, numerous applications are moving towards practicality, such as quantum computing \cite{Google2019Sycamore,USTC2020Jiuzhang,USTC2021zuchongzhi}, quantum communication \cite{Cozzolino.etal2019,Bhaskar.etal2020}, and quantum metrology \cite{Pezze.etal2018,Polino.etal2020}. 
Meanwhile, a fundamental task in all of these applications, i.e., quantum characterization is becoming more and more crucial alongside. 
The standard tool of quantum tomography \cite{QSE2004} is powerful, but the exponential increasing of quantum resource and data processing time \cite{Haffner.etal2005, Shang.etal2017} makes the characterization of large quantum systems extremely tricky.
Hence, attention has been turned to nontomographic methods \cite{Mayers.etal2004,Toth.Guehne2005c,Guehne.Toth2009,Flammia.Liu2011,Dimic.Dakic2018}, among which quantum state verification (QSV) \cite{Pallister.etal2018} particularly stands out due to its unconditionally high efficiency, which is at least quadratically better than other methods. 
In the last few years, certain types of entangled quantum states and processes have been verified efficiently or even optimally using local measurements only \cite{Morimae.etal2017, Takeuchi.Morimae2018, Yu.etal2019, Li.etal2019, Wang.Hayashi2019, Zhu.Hayashi2019a, Zhu.Hayashi2019b, Zhu.Hayashi2019c, Zhu.Hayashi2019d, Liu.etal2019b, Li.etal2020b, Dangniam.etal2020, Zhang.etal2020a, Jiang.etal2020, Zhang.etal2020b, Li.etal2020a, Liu.etal2020b,Liu.etal2021, Han.etal2021, Zhu.etal2022, ChenHuang.etal2023, Li.Y.etal2021, Liu.etal2020a, Zhu.Zhang2020, Zeng.etal2020, Zhang.etal2022}; see Ref.~\cite{Yu.etal2022} for a recent review.

The task of QSV is to verify a quantum device which is supposed to produce the target state $\ket{\psi}$, but in fact states $\varrho_1,\varrho_2,\cdots,\varrho_N$ might be emitted. 
To accomplish this mission, one can randomly perform some pass-or-fail tests $\{\Omega_i,\openone-\Omega_i\}$, which detect any bad state $\varrho_j$ with fidelity ${\bra{\psi}\varrho_j\ket{\psi} \leq 1-\epsilon}$. 
Then, a verification protocol can be constructed as ${\Omega = \sum_i \mu_i\Omega_i}$, where $\{\mu_i\}$ denotes a probability distribution. 
Importantly, these tests should be suitably designed such that the target state can always pass, i.e., ${\Omega_i\ket{\psi}=\ket{\psi}}, \forall i$.
The maximal probability that a bad state passes the protocol is ${1-\epsilon \nu(\Omega)}$, where ${\nu(\Omega):=1-\lambda_2(\Omega)}$ denotes the spectral gap between the largest and the second largest eigenvalues of $\Omega$ \cite{Pallister.etal2018, Zhu.Hayashi2019c}. 
Thus, in order to gain a confidence level ${1-\gamma}$, the protocol $\Omega$ requires
\begin{equation}\label{eq:QSVparameter}
  N\geq\frac{\ln\gamma^{-1}}{\ln\bigl\{[1-\nu(\Omega)\epsilon]^{-1}\bigr\}}\approx
  \frac1{\nu(\Omega)}\epsilon^{-1}\ln\gamma^{-1}\,
\end{equation}
copies of the state to verify $\ket{\psi}$ within infidelity $\epsilon$. 

Similarly, other characterization methods including tomography can also be performed probabilistically with the measurement settings following a probability distribution.
The unknown state can then be reconstructed by tomography once the expectation values of all the measurement settings (or observables) are obtained, the success of which demands that the measurement outcomes and the settings must be matched. 
QSV, instead, owns a distinct feature that the measurement protocol can be regarded as a black box as we only need to know the numbers of ``pass'' and ``fail'' outcomes.
The one-to-one correspondence between the measurement outcomes and the settings is not necessary in QSV.
Hence, we dub such measurement protocols as being \textit{choice-independent}.

\begin{figure*}[t]
  \includegraphics[width=0.8\textwidth]{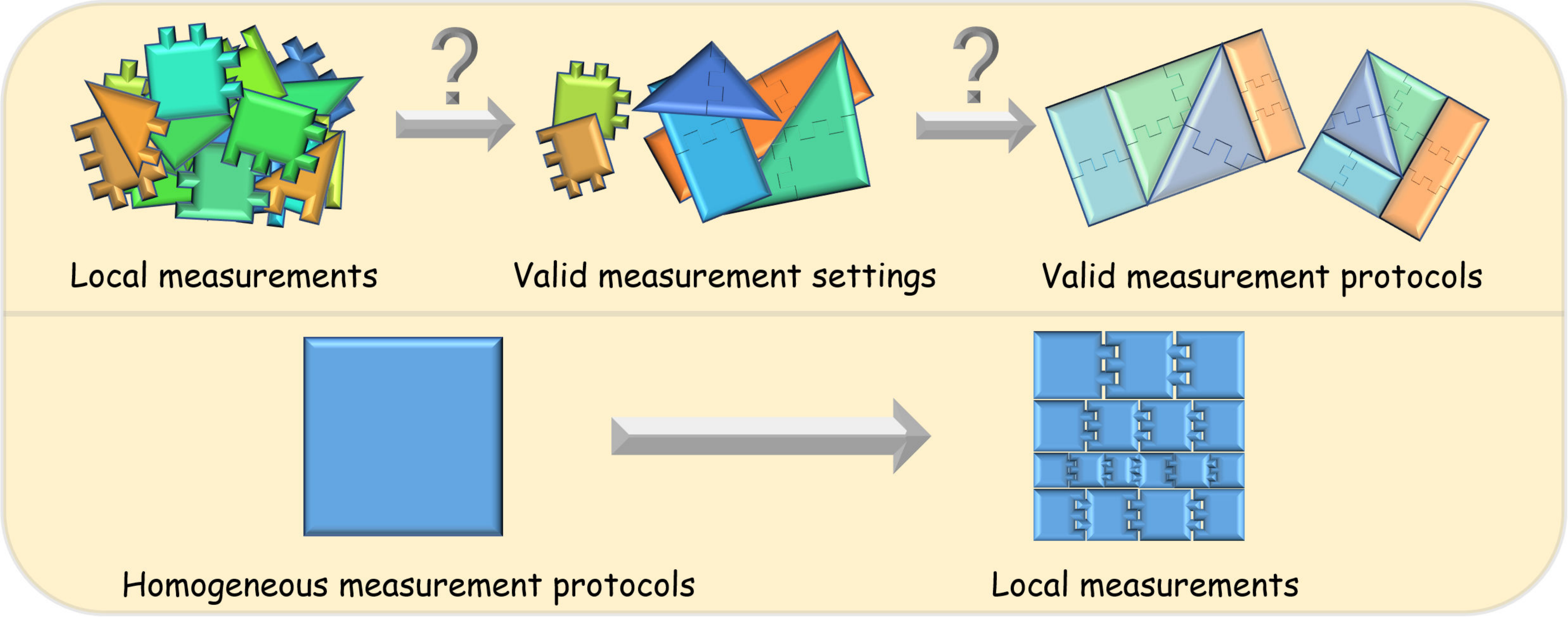}
  \caption{
    Schematic illustration of the protocol design. Instead of searching for valid measurement settings then protocols among numerous possible local measurements, one can start with homogeneous measurement protocols directly by checking whether they can be local or not.
    }
  \label{fig:scheme}
\end{figure*}

In general, constructing an efficient verification protocol for a target state with only local measurements is difficult.
Nevertheless, as the verification protocols are choice-independent, it is likely to start from an overall perspective of the verification protocol $\Omega$. 
In this way, the problem can be converted to check whether $\Omega$ can be realized locally or not. 
Coincidentally, in the study of QSV, the so-called \textit{homogeneous} protocols \cite{Zhu.Hayashi2019c,Zhu.Hayashi2019d} emerge, of which $\Omega$ can be directly written down. 
The structure of homogeneous verification protocols is highly symmetric such that they have the best performance in the adversarial scenario \cite{Zhu.Hayashi2019c,Zhu.Hayashi2019d}. 
Indeed, several types of entangled states have already been efficiently verified by local homogeneous measurements \cite{Pallister.etal2018,Li.etal2019,Zhu.Hayashi2019a,Li.etal2020b,Li.etal2020a}.
Therefore, the equivalence between protocol design and locality of homogeneous measurements offers us a possible way out for the open problem of verifying arbitrary entangled states; see 
Fig.~\ref{fig:scheme}. 
In fact, such an idea appears in several other problems including the construction of entanglement witness \cite{Bourennane.etal2004} and parameter estimation \cite{Brandao.etal2005,Li.etal2020b}, whose protocols can also be choice-independent.

In this work, we propose a systematic strategy to design QSV protocols for arbitrary entangled pure states. 
The main idea is to convert this problem to the checking of the locality of homogeneous protocols. 
First, we formalize the locality of a measurement protocol and answer the basic question of whether it is local. 
Focusing on local measurements, we derive the corresponding constraints explicitly.
Then, the more general scenario of using local operations and classical communication (LOCC) as well as the case of infinite continuous measurements are discussed respectively.
Next, for the homogeneous QSV protocols whose operators can be directly achieved for arbitrary states, we acquire the corresponding constraints for them being local.
Specifically, we demonstrate the explicit procedures of the protocol design using Pauli projections, for verifying Bell states, stabilizer states (GHZ states in particular) and $W$ states. 
For certain cases, our method attains the best strategies to date.
Finally, our framework is successfully extended to other tasks including the construction of entanglement witness and parameter estimation.

\section{Locality of measurement protocols}
From a more general perspective, consider an arbitrary measurement protocol $\Pi$, which can be decomposed into
\begin{eqnarray}\label{eq:protocol}
  \Pi=\sum_i \mu_i \Pi_i\,,
\end{eqnarray}
where $\Pi_i$s are individual measurement settings and $\{\mu_i\}$ is a probability distribution. 
We would like to find out under what circumstances $\Pi$ is local.
Without loss of generality, assume that we can realize an ensemble of $s$-outcome positive operator-valued measures (POVMs) ${\{M_{i}^1,\cdots,M_{i}^s\}_{i}}$, where ${\sum_{j=1}^s M_{i}^j=\openone}$.
For an $n$-partite quantum system, the protocol $\Pi$ is called local if all the measurement settings are local, such that
\begin{eqnarray}\label{eq:setting}
  \Pi_{i=i_1\cdots i_n}=\sum_{j} z_{i,j}M_{i_1}^{j_1}\otimes\cdots\otimes M_{i_n}^{j_n}\,,\quad
\end{eqnarray}
where the sum is taken over ${j\!=\!j_1\cdots j_n\!\in\!\{1,\cdots,s\}^{\otimes n}}$, and the parameters $z_{i,j}$ are either $0$ or $1$ that tell us which outcomes ${j_1\cdots j_n}$ of the measurement setting $\Pi_i$ correspond to the ``pass'' instances. 
More generally, we can let ${z_{i,j}\in[0,1]}$ if some outcomes are allowed to pass the test with probability ${0<z_{i,j}<1}$. 

By combining all the measurement settings, we get the decomposition of a measurement protocol with the form
\begin{eqnarray}\label{eq:sum}
  \Pi = \sum_{i,j} p_{i,j} M_{i_1}^{j_1} \otimes \cdots \otimes M_{i_n}^{j_n}\,,
\end{eqnarray}
where ${p_{i,j}(\Pi):=\mu_i z_{i,j}}$ is called the quasi-probability distribution, as $\sum_{i,j} p_{i,j}(\Pi) $ is typically not equal to $1$.
With this, we have the following theorem for the locality of measurement protocols.
\begin{theorem}\label{thm:measure_d}
A measurement protocol is local \emph{iff} the quasi-probability distribution $p_{i,j}$ satisfies the following two constraints under the representation of Eq.~\eqref{eq:sum}, 
\begin{align}
    &\bullet{\bf Positivity}&&\min_{i,j}\bigl\{p_{i,j}(\Pi)\bigr\} \geq 0\,,\label{eq:pos_d}\\
    &\bullet{\bf Completeness}&&S(\Pi):= \sum_{i} \max_{j} \bigl\{p_{i,j}(\Pi)\bigr\} \leq 1\,.\label{eq:cop_d}
\end{align}
\end{theorem} 
\begin{proof}
  For a measurement protocol $\Pi$ as in Eq.~\eqref{eq:protocol}, the probability distribution should satisfy (1) positivity ${\mu_i\geq 0,\forall i}$; and (2) completeness ${\sum_i\mu_i=1}$. 
  From Eq.~\eqref{eq:sum}, the quasi-probability distribution is given by ${p_{i,j}=\mu_i z_{i,j}}$ with ${z_{i,j}\in [0,1]}$, which then leads to the two constraints straightforwardly.
\end{proof}

Our consideration here can be naturally extended to two other perspectives. First, one can consider local operations and classical communication (LOCC) for the measurement settings, known also as adaptive measurements \cite{Yu.etal2019, Li.etal2019, Wang.Hayashi2019, Liu.etal2019b}; see Appendix~\ref{AppendixA} for the corresponding discussions. 
Second, the decomposition of $\Pi$ as in Eq.~\eqref{eq:sum} with finite local measurements can be generalized to the infinite scenario using continuous local measurements; see Appendix~\ref{AppendixB} for detailed discussions.

\section{Homogeneous QSV protocols with finite local projections}
A first application of the previous discussion on the locality of measurement protocols is QSV.
To verify a target pure state $\ket{\psi}$, a homogeneous QSV protocol takes on the general form
\begin{eqnarray}\label{eq:def_homo}
  \Omega_{\text{\rm Hom}} = (1-\nu) \openone + \nu \ket{\psi}\bra{\psi}\,,
\end{eqnarray}
where $0<\nu\leq 1$.
All eigenvalues of the homogeneous protocol are $1-\nu$ except the largest one which is the unity.
The parameter $\nu$ is exactly the spectral gap, thus $1/\nu$ gives the scaling of the verification efficiency for homogeneous protocols.

We note that the identity $\openone$ is a trivial measurement, which can be considered as no measurement at all or an arbitrary measurement whose outcomes must be accepted. 
Thus for the limit ${\nu=0}$, the protocol ${\Omega_{\text{\rm Hom}}=\openone}$ must be local.
On the other hand, the protocol becomes ${\Omega_{\text{\rm Hom}}=\ket{\psi}\bra{\psi}}$ for ${\nu=1}$ which cannot be realized locally since the target state $\ket{\psi}$ is assumed to be entangled.
Hence, in general, the homogeneous QSV protocol $\Omega_{\text{\rm Hom}}$ as in Eq.~\eqref{eq:def_homo} is the convex combination of the identity $\openone$ and the projection $\ket{\psi}\bra{\psi}$, and the locality of $\Omega_{\text{\rm Hom}}$ can be interpreted as finding the \textit{local ball} around $\openone$ with the maximal value of the parameter $\nu$ representing the radius.

An arbitrary measurement protocol for $n$-qubit systems can always be expanded with the Pauli representation uniquely as 
\begin{eqnarray}\label{eq:Pauli_Rep}
  \Pi = \frac{1}{2^n} \sum_{\alpha} c_{\alpha} \sigma_{\alpha_1} \otimes \cdots \otimes \sigma_{\alpha_n}\,,
\end{eqnarray}
where the coefficients are ${c_{\alpha}=\tr(\Pi \sigma_{\alpha_1} \otimes \cdots \otimes \sigma_{\alpha_n})}$ with ${\alpha \!=\! \alpha_1 \cdots \alpha_n  \!\in\! \{0,1,2,3\}^{\otimes n}}$, and ${\sigma_0=\openone}$. 
Expanding the Pauli operators with a finite set of measurements $\{M_i^1,\cdots,M_i^s\}_{i=1}^{k}$, we have
\begin{eqnarray}
  \sigma_\alpha = \vec{t}_\alpha \cdot \vec{M}\,,
\end{eqnarray}
where ${\vec{t}_\alpha\!=\!\bigl[t_\alpha^{(1,1)},\cdots,t_\alpha^{(1,s)},\cdots,t_\alpha^{(k,1)},\cdots,t_\alpha^{(k,s)}\bigr]}$ and ${\vec{M} = [M_1^1,\cdots,M_1^s,\cdots,M_k^1,\cdots,M_k^s]^T}$. 
Then the Pauli representation can be converted to
\begin{eqnarray}
  \sigma_{\alpha_1} \otimes \cdots \otimes \sigma_{\alpha_n} 
  = \vec{t}_{\alpha_1} \otimes \cdots \otimes \vec{t}_{\alpha_n} \cdot \vec{M}^{\otimes n}\,. 
\end{eqnarray}
Thus, we get the quasi-probability distribution 
\begin{eqnarray}\label{eq:Rep_Proj}
  p_{i_1\cdots i_n}
  := 
  \biggl(
  \frac{1}{2^n}\sum_{\alpha} c_{\alpha}
  \vec{t}_{\alpha_1} \otimes \cdots \otimes \vec{t}_{\alpha_n}
  \biggr)_{i_1\cdots i_n}\,.\quad
\end{eqnarray}
Notice that due to the flattening vector $\vec{M}$, the quasi-probability $p_{i_1\cdots i_n}$ is the flattening one-dimensional vector from the previous $p_{i,j}$ as defined in Theorem~\ref{thm:measure_d}.
Such a quasi-probability distribution is a linear function, then for the homogeneous QSV protocol, we have 
\begin{eqnarray}\label{eq:pro_homP}
  p_i(\Omega_\text{Hom})
  = (1-\nu)p_i(\openone)
  + \nu p_i(\psi)\,,
\end{eqnarray}
where $p_i(\openone)$ and $p_i(\psi)$ are the quasi-probability distributions of the identity and the projection on the target state, respectively.

The transformation between Pauli operators and the finite set of measurements $\{M_i^1,\cdots,M_i^s\}_{i=1}^{k}$ requires these measurements to constitute a complete set of bases in the Hilbert space.
In the following, we consider the standard Pauli projections, which are easy to realize in experiments.

\section{Homogeneous QSV protocols with Pauli projections}
The Pauli projections $\{P_i^0,P_i^1\}_{i=1}^{3}$ form an overcomplete set of bases for qubit systems, thus the representation is not unique.
One possible transformation is $P_{i}^{j}=\frac{1}{2}\bigl[\openone + (-1)^j \sigma_i\bigr]$ for $i=1,2,3$.
As for ${\sigma_0=\openone}$, we specifically choose the symmetric form $\openone=\frac{1}{3} \sum_{i,j} P_i^j$, i.e.,
\begin{eqnarray}\label{eq:def_T}
  T_\text{Pauli}:=
  \begin{bmatrix}
    \vec{t}_0\\
    \vec{t}_1\\
    \vec{t}_2\\
    \vec{t}_3
  \end{bmatrix}
  =
  \begin{bmatrix}
    1/3   &1/3    &1/3    &1/3    &1/3    &1/3  \\
    1     &-1     &0      &0      &0      &0    \\
    0     &0      &1      &-1     &0      &0    \\
    0     &0      &0      &0      &1      &-1   \\
  \end{bmatrix}
  \!.\quad
\end{eqnarray}
With such a transformation, the following corollary can be generated from Theorem~\ref{thm:measure_d}; see Appendix~\ref{AppendixC} for the proof.
\begin{corollary}\label{cor:Pauli}
  Considering a homogeneous QSV protocol $\Omega_{\text{\rm Hom}}$ for the target state \ket{\psi} as defined in Eq.~\eqref{eq:def_homo}, it is local under Pauli projections if the following two constraints are satisfied:
\begin{align}
  &\bullet{\bf Positivity}\nonumber\\
  &\qquad\quad\nu \leq \frac{1}{1- 3^n \min \{p_i(\psi)\}}
    \leq \frac{1}{2^n\!-\!2^{1-\!n}\!+\!1}\,.\hfill\label{eq:pos_homP}\\
  &\bullet{\bf Completeness} \qquad S(\psi) \leq 1\,.\label{eq:cop_homP}
\end{align}
\end{corollary} 

Two remarks are in order. 
First, for the positivity constraint, we note that the first inequality in Eq.~\eqref{eq:pos_homP} gives the radius of the local ball;  
while the second one is obtained by considering all possible target states.
In other words, for the homogeneous QSV protocol, any target state can be verified with an efficiency no more than $O(2^n)$ as long as the completeness constraint is satisfied. 
Second, here we focus on the locality of homogeneous QSV protocols with local Pauli projections only, and it is reasonable to expect that more measurements should improve the efficiency of the protocols.
However, this is not true as the efficiency is still bounded by $O(2^n)$ even infinite continuous local projections or multi-outcome measurements are considered; see Appendixes \ref{AppendixB2} and \ref{AppendixB3} for more details.

\section{Applications}
By employing Pauli projections, here we consider the verification of several typical entangled states with our method.
More detailed analyses can be found in Appendix~\ref{AppendixD}.

(i) Bell state ${\ket{\Phi^{+}}=\frac1{\sqrt{2}}\bigl(\ket{00}+\ket{11}\bigr)}$.
By using our method, a homogeneous QSV protocol using only local Pauli measurement settings can be designed with an efficiency of ${1/\nu=3}$. 
Clearly, such a protocol is not optimal as the best one has an efficiency of ${1/\nu=3/2}$ \cite{Pallister.etal2018}. 
The reason lies in that a specific quasi-probability distribution is chosen in our method where the identity operation exists in each measurement setting. 
Hence, with an appropriate revision process, our protocol can be improved to give exactly the optimal efficiency.

(ii) Three-qubit GHZ state $\ket{\text{GHZ}_3}=\frac1{\sqrt{2}}\bigl(\ket{000}+\ket{111}\bigr)$.
With our method, the local homogeneous QSV protocol designed has an efficiency of ${1/\nu=17/4}$.
It is worse than that in Ref.~\cite{Pallister.etal2018} which is ${1/\nu=7/4}$.
However, with an additional revision process, a better efficiency of ${1/\nu=5/3}$ can be obtained. 
Moreover, one can achieve the optimal efficiency of ${1/\nu=3/2}$ \cite{Li.etal2020b} with a proper choice of the transformation $T_\text{Pauli}$.
Note that all stabilizer states can be verified by QSV protocols constructed with their stabilizers which are in the Pauli group \cite{Pallister.etal2018}, thus our method is able to give the local homogeneous QSV protocols for all stabilizer states with the quasi-probability distribution based on the Pauli representation. 

(iii) Three-qubit $W$ state $\ket{W_3}=\frac1{\sqrt{3}}\bigl(\ket{001}+\ket{010}+\ket{100}\bigr)$. 
In this case, our method is not able to give a local homogeneous QSV protocol as the quasi-probability distribution has ${S(W_3)=1.40(7)}$ which violates the completeness constraint. 
Even with a revision process, the constraint ${\tr(\Omega_i\ket{\psi}\bra{\psi})=1}$ for QSV cannot be satisfied for all the settings.
Note, however, that this only means $\ket{W_3}$ cannot be verified by any local homogeneous protocol using the Pauli projections, but with other local projections it might be possible.
On the other hand, a valid inhomogeneous protocol can be achieved by properly choosing the settings.
Such a protocol has an efficiency of ${1/\nu=13/3}$, which is better than that of the previous inhomogeneous protocol with Pauli projections \cite{Liu.etal2019b}.

More importantly, if we allow LOCC for the verification protocol, we are able to obtain a local homogeneous QSV protocol by modifying the completeness constraint.
The efficiency is given by ${1/\nu=2}$, which is better than the ${1/\nu=3}$ reported in Ref.~\cite{Liu.etal2019b}.
It is slightly worse than that of the nearly optimal homogeneous protocol of ${1/\nu=8/5}$ \cite{Li.etal2020a}, which however, requires much more complex local measurement settings.

\section{Extended applications}
The abstraction of choice-independent measurement protocols can be naturally extended to other tasks concerning only the protocol operators instead of the specific settings, such as the entanglement witness for detecting entanglement.
Moreover, with appropriate modifications,  the universality of choice-independent measurement protocols enables its extension to the task of parameter estimation including fidelity, entanglement and so on. 

\subsection{Construction of entanglement witness}
An entanglement witness $W$ is defined if for every separable state $\rho_\text{sep}$, one has $\tr(W\rho_\text{sep})\geq 0$; and for some entangled state $\rho_\text{ent}$, $\tr(W\rho_\text{ent})<0$.
Witnesses for detecting entanglement are typically of the form
\begin{eqnarray}
  W=\kappa \openone - \ket{\psi}\bra{\psi}\,,
\end{eqnarray}
where \ket{\psi} is the entangled state to be detected.
The parameter $\kappa$ is the square of the maximal Schmidt coefficient of \ket{\psi} when all bipartitions are considered \cite{Bourennane.etal2004}.

We can associate entanglement witnesses with the homogeneous protocols as 
\begin{eqnarray}
  W=\biggl(\kappa+\frac{1-\nu}{\nu}\biggr)\openone-\frac{1}{\nu}\Omega_\text{Hom}\,.
\end{eqnarray}
Hence, to determine whether a state $\ket{\psi}$ is entangled is equivalent to verifying if the target state is $\ket{\psi}$ within the infidelity ${\epsilon=1-\kappa}$. 
Such a relation transforms the witnesses from the formation of observables to the construction of choice-independent measurement protocols.
This equivalence improves the estimation of shot noise from $1/\sqrt{N}$ (statistical mean error) to $1/N$ (error of hypothesis testing).

\subsection{Parameter estimation}
Considering choice-independent measurement protocols, one finds that homogeneous QSV can also be regarded as fidelity estimation \cite{Li.etal2020b}, i.e., 
\begin{eqnarray}
  F=\bra{\psi}\sigma\ket{\psi}= \frac{\tr(\Omega_\text{Hom}\sigma)-(1-\nu)}{\nu}\,,
\end{eqnarray}
with standard deviation 
\begin{eqnarray}
  \Delta F = \frac{\sqrt{(1-F)(F+\nu^{-1}-1)}}{N}\leq\frac{1}{2\nu\sqrt{N}}\,.
\end{eqnarray}
It shows that the number of copies required is ${N\sim O(\epsilon^2)}$, which is worse than that of verification. 
Nevertheless, one can directly achieve the value of fidelity rather than a bound.
In addition, performing fidelity estimation only needs to know the frequency of pass instances rather than the number of successive ones, which is much more robust in experiments. 
Moreover, considering entanglement quantified by witness operators \cite{Brandao.etal2005}, the local protocol for entanglement estimation can be similarly designed as being choice-independent as well.

For a homogeneous protocol $\Omega_\text{Hom}$, the positivity constraint can always be achieved with a proper $\nu$. 
If the completeness constraint is violated, we can consider the measurement protocol
\begin{eqnarray}
 \tilde{\Omega}_\text{Hom}= \frac{\Omega_\text{Hom}}{S(\Omega_\text{Hom})}\,.
\end{eqnarray}
Now the passing probability of the target state is given by $1/S(\Omega_\text{Hom})$, so it cannot be used for verification.
However, the task of estimation is immune to this problem as we only need to add a corresponding factor of scaling. 
Exemption of the completeness constraint enables our method to give a local measurement protocol for arbitrary estimation tasks, which is comparable to the optimal one.

\section{Conclusion}
We have proposed a systematic strategy to design QSV protocols for arbitrary entangled pure states. 
By initiating the concept of choice-independent measurement protocols, we have successfully converted the original problem to the checking of the locality of homogeneous protocols. 
By formalizing the locality of measurement protocols, we have derived the corresponding constraints for local measurements, LOCC, as well as infinite continuous measurements respectively.
Then, for the homogeneous QSV protocols whose operators can be directly written down for arbitrary pure states, we acquired the corresponding constraints for them being local. 
Specifically, we demonstrated the explicit procedures of the protocol design using Pauli projections, for verifying Bell states, stabilizer states and $W$ states.
For certain cases, our method
has achieved the best strategies to date.

Furthermore, the discussions on the locality of measurement protocols can be applied to more tasks such as the construction of entanglement witness. 
Finally, we have shown that all these tasks can be converted to parameter estimation.
In this case, the local measurement protocols can be directly given, as the constraints of local protocols for these tasks can always be satisfied.

\acknowledgments
We are grateful to Otfried G\"uhne, Huangjun Zhu and Zihao Li for helpful discussions.
This work was supported by the National Natural Science Foundation of China (Grants No.~12175014, No.~92265115, and No.~91850205) and the National Key R\&D Program of China (Grant No.~2022YFA1404900). 
Y.-C.L. is also supported by the Deutsche Forschungsgemeinschaft (DFG, German Research Foundation, project numbers 447948357 and 440958198) and the Sino-German Center for Research Promotion (Project M-0294).

%
%

\onecolumngrid

\appendix

\section{Local operations and classical communication}\label{AppendixA}
In the main text, local measurements are employed for the construction of measurement protocols.
Here, we extend our scheme by considering local operations and classical communication (LOCC), known also as adaptive measurements. 
For an $n$-partite quantum system, a measurement protocol with adaptive measurements can be written as
\begin{eqnarray}\label{eq:protocol_adp}
    \Pi = \sum_{k}\mu_k \Pi_k\,,
\end{eqnarray}
where $\Pi_k$s represent adaptive measurement settings constructed by $n$ local $s$-outcome POVMs ${\{M_{i}^1,\cdots,M_{i}^s\}_{i}}$ with different measurement orders. It can be understood as
\begin{eqnarray}
    \Pi_k 
    &=& \cP_k\Bigl\{\sum_{i,j}q_k(i_1)q_k(i_2|i_1 j_1)q_k(i_3|i_1 j_1,i_2 j_2)\cdots q_k(i_n|i_1 j_1,i_2 j_2,\cdots, i_{n-1} j_{n-1}) \nonumber\\
    && z_k(j_n|i_1 j_1,i_2 j_2,\cdots, i_{n-1} j_{n-1})
    M_{i_1}^{j_1} \otimes \cdots \otimes M_{i_n}^{j_n}\Bigr\}\nonumber\\
    &=& \sum_{i,j}q_{i,j,k} M_{i_1}^{j_1} \otimes \cdots \otimes M_{i_n}^{j_n}\,.
\end{eqnarray}
The setting is adaptively measured from $M_{i_1}$ to $M_{i_n}$ with the probability distribution $q_k(\cdot|\cdot)$ based on previous outcomes, and all possible measurement orders are considered with the permutation operator $\cP_k\{\cdot\}$.  
The parameter $z_k(\cdot|\cdot)$ is either $0$ or $1$ that tells us which outcomes of the finial measurement correspond to ``pass'' instances, and it can be generally considered as $z_k\in[0,1]$.
Denoting them compactly as $q_{i,j,k}$, we have
\begin{eqnarray}\label{eq:sum_adp}
  \Pi = \sum_{i,j} p_{i,j} M_{i_1}^{j_1} \otimes \cdots \otimes M_{i_n}^{j_n}\,,
\end{eqnarray}
where $p_{i,j}:=\sum_{k}\mu_k q_{i,j,k}$.

\begin{theorem}\label{thm:measure_a}
A measurement protocol can be realized with LOCC \emph{iff} the quasi-probability distribution $p_{i,j}$ satisfies the following two constraints under the representation of Eq.~\eqref{eq:sum_adp}, 
\begin{align}
    &\bullet{\bf Positivity}&& \min_{i,j}\bigl\{p_{i,j}(\Pi)\bigr\} \geq 0\,,\label{eq:pos_adp}\\
    &\bullet{\bf Completeness}&& S(\Pi):= \sum_{i} \max_{j} \bigl\{p_{i,j}(\Pi)\bigr\} \leq s^{n-1}\,.\label{eq:cop_adp}
\end{align}
\end{theorem} 
\begin{proof}
  For the measurement protocol $\Pi$ as in Eq.~\eqref{eq:protocol_adp}, the probability distribution should satisfy (1) positivity ${\mu_k\geq 0,\forall k}$; and (2) completeness ${\sum_k\mu_k=1}$. 
  Then the positivity constraint in Eq.~\eqref{eq:pos_adp} can be directly achieved. For the completeness constraint, considering only one adaptive measurement setting, there are $\sum_{i_1}q_k(i_1)=1$ and $\sum_{i_2}q_k(i_2|i_1 j_1)=1$. Then we have $\sum_{i_1}q_k(i_1)q_k(i_2|i_1 j_1)\leq 1$, and thus $\sum_{j_1}\sum_{i_1}q_k(i_1)q_k(i_2|i_1 j_1)\leq s$. 
  With $z_k(j_n|\cdot)\in[0,1], \forall j_n$, for each adaptive measurement setting, we have $\sum_{i} \max_{j} \bigl\{q_{i,j,k}(\Pi)\bigr\}\leq s^{n-1}\cdot 1$. Finally, 
  for the quasi-probability distribution ${p_{i,j}=\sum_{k}\mu_k q_{i,j,k}}$, one can deduce $\sum_{i} \max_{j} \bigl\{p_{i,j}(\Pi)\bigr\} \leq s^{n-1}$.
\end{proof}

\section{Extension to infinite continuous local measurements}\label{AppendixB}
\subsection{Locality of measurement protocols based on infinite continuous local measurements}
Considering $n$-qubit systems, the decomposition of $\Pi$ as in Eq.~\eqref{eq:sum} with finite local measurements can be generalized to the infinite scenario using continuous local projections over the Bloch sphere, such that
\begin{eqnarray}\label{eq:int}
  \Pi=\int_\cB \D\cV_1\cdots \D\cV_n w_{\vec{r}_1,\cdots,\vec{r}_n}(\Pi) P_{\vec{r}_1}\otimes\cdots\otimes P_{\vec{r}_n}\,,
\end{eqnarray}
where $\cB$ denotes the integral over $n$ Bloch spheres.
The local operator ${P_{\vec{r}} = \frac{1}{2}(\openone+\vec{r}\cdot\vec{\sigma})}$ is the projection onto the pure state located at the unit vector $\vec{r}$, and ${\vec{\sigma}=(\sigma_1,\sigma_2,\sigma_3)}$ are the Pauli matrices.
Then, one obtains the following theorem.
\begin{theorem}\label{thm:measure_c}
A measurement protocol for $n$-qubit systems is local \emph{iff} the quasi-probability distribution $w_{\vec{r}_1,\cdots,\vec{r}_n}(\Pi)$ satisfies the following two constraints under the diagonal $P$-representation of Eq.~\eqref{eq:int},
\begin{align}
    &\bullet{\bf Positivity}&&\min_{\forall \vec{r}}\bigl\{w_{\vec{r}_1,\cdots,\vec{r}_n}(\Pi)\bigr\} \geq 0 \,,\qquad\quad\label{eq:pos_mea}\\
    &\bullet{\bf Completeness}&&\cS(\Pi) \leq 2^n\,.\label{eq:cop_mea}
\end{align}
 $\cS(\Pi)$ is the integral over the envelope surface of all $2^n$ quasi-probability distributions $\bigl\{w_{\vec{r}_{1,j_1},\cdots,\vec{r}_{n,j_n}}(\Pi)\bigr\}_{j\in\{0,1\}^{\otimes n}}$, i.e.,
  \begin{equation}
    \cS(\Pi):=\int_\cB \D\cV_1\cdots \D\cV_n \tilde{w}_{\vec{r}_1,\cdots,\vec{r}_n}(\Pi)\,,
  \end{equation}
  where ${\tilde{w}_{\vec{r}_1,\cdots\!,\vec{r}_n}(\Pi)
    =
    \displaystyle\max_{j\in\{0,1\}^{\otimes n}}
    \bigl\{w_{\vec{r}_{1,j_1},\cdots\!,\vec{r}_{n,j_n}}\bigr\}}$, with ${\vec{r}_{i,0}:=\vec{r}_i}$ and ${\vec{r}_{i,0}+\vec{r}_{i,1}=0}$. 
\end{theorem}
\begin{proof}
  Similar to the proof of Theorem~\ref{thm:measure_d}, the quasi-probability distribution can be considered as the multiplication of the probability distribution of measurement settings and their corresponding pass probabilities, i.e., ${w_{\vec{r}_1,\cdots,\vec{r}_n}=\mu_{\vec{r}_1,\cdots,\vec{r}_n}z_{\vec{r}_1,\cdots,\vec{r}_n}}$ with ${z\in[0,1]}$. 
  Note that the subscript is omitted if no ambiguity arises.
  The probability distribution should satisfy (1) positivity ${\mu\geq 0, \forall \vec{r}}$, so Eq.~\eqref{eq:pos_mea} is directly achieved; (2) completeness ${\int_{|\cB|} \D\cV \mu = 1}$, where $|\cB|$ represents the integral over the space of all measurement settings. 
  Note that for $n$-qubit systems, the projective measurements are  $\bigl\{P_{\vec{r}_{i,0}},P_{\vec{r}_{i,1}}\bigr\}$ with ${\vec{r}_{i,0}+\vec{r}_{i,1}=0}$. 
  Thus, the assemble of projections $\bigl\{P_{\vec{r}_{1,j_1}}\otimes\cdots\otimes P_{\vec{r}_{n,j_n}}\bigr\}_{j\in\{0,1\}^{\otimes n}}$ are in the same setting. Then, with the integral over $n$ Bloch spheres, we have ${\int_{\cB} \D\cV \mu = 2^n}$. 
  As ${z\leq 1}$, we get $\mu_{\vec{r}_1,\cdots,\vec{r}_n} \leq \max_{j\in\{0,1\}^{\otimes n}}
    \{\!w_{\vec{r}_{1,j_1},\cdots\!,\vec{r}_{n,j_n}}\}:=\tilde{w}_{\vec{r}_1,\cdots,\vec{r}_n}$, and finally Eq.~\eqref{eq:cop_mea} follows.
    See Fig.~\ref{fig:envelope} for an illustration of the quasi-probability distribution for two-qubit systems.
\end{proof}

\begin{figure}[t]
  \includegraphics[width=0.4\textwidth]{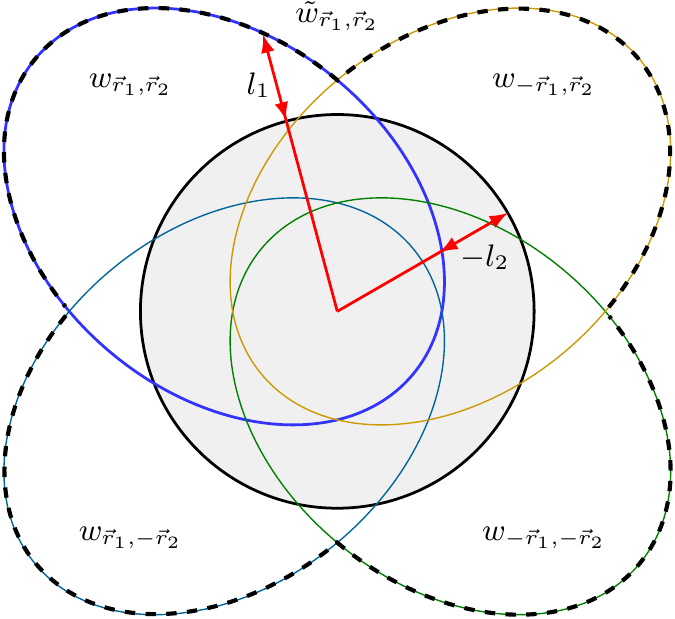}
  \caption{
    Illustration of the quasi-probability distribution for two-qubit systems. 
    The circle in black denotes the Hilbert space $\mathcal{H}^{\otimes 2}$, and the ellipse in blue is the quasi-probability distribution $w_{\vec{r}_1,\vec{r}_2}$. 
    The distance between a point in the Hilbert space and another point in the quasi-probability distribution represents the value of $w_{\vec{r}_1,\vec{r}_2}$.
    If the points are outside of the Hilbert space, the distance is positive ($l_1$); otherwise it is negative ($l_2$). 
    The envelope $\tilde{w}_{\vec{r}_1,\vec{r}_2}$ is shown with dashed curves.
    Note that the intersection of the circle and the ellipse indicates that  it is against the positivity constraint as $w<0$ for some $\vec{r}$. Hence the illustration here is not a quasi-probability distribution for local measurement protocols.
    }
  \label{fig:envelope}
\end{figure}

Several remarks are in order. 
First, the inequality of the completeness constraint is saturated if the protocol is trivial as $\openone$ such that ${\tilde{w}(\openone)=w(\openone)}$.
Second, for the infinite scenario, if we consider $s$-outcome projections like the symmetric informationally complete (SIC) POVMs, the upper bound of the completeness constraint changes from $2^n$ to $s^n$.
Finally, if $w_{\vec{r}_1,\cdots,\vec{r}_n}=\sum_{i,j}p_{i,j}\delta(\vec{r}_1-\vec{r}_{i_1,j_1})\cdots\delta(\vec{r}_n-\vec{r}_{i_n,j_n})$ with $\delta(\cdot)$ being the delta function, Theorem~\ref{thm:measure_c} reduces to Theorem~\ref{thm:measure_d} for qubit systems.

\subsection{Homogeneous QSV protocols with infinite local projections}\label{AppendixB2}
We revisit the homogeneous QSV protocol in its general form as in Eq.~\eqref{eq:def_homo},
\begin{eqnarray}
  \Omega_{\text{\rm Hom}} = (1-\nu) \openone + \nu \ket{\psi}\bra{\psi}\,,\quad(0<\nu\leq 1)
  \label{eq:def_homoA}
\end{eqnarray}
for a target pure state \ket{\psi}.
We consider the locality of $\Omega_{\text{\rm Hom}}$ with infinite local projections by following Theorem~\ref{thm:measure_c}.
Note that the quasi-probability distribution $w(\Pi)$ is not unique for the measurement protocol $\Pi$. 
Considering the spherical harmonics expansion of $w(\Pi)$ with order-$0$ and order-$1$ components only, a unique representation for $n$-qubit systems can be written as
\begin{eqnarray}\label{eq:def_w}
  w_{\vec{r}_1,\cdots,\vec{r}_n}(\Pi)
  =\frac{1}{(4\pi)^n}
  \tr\bigl[\Pi (\openone+3\vec{r}_1\cdot\vec{\sigma})\otimes\cdots\otimes(\openone+3\vec{r}_n\cdot\vec{\sigma})\bigr]\,,
\end{eqnarray}
and higher-order spherical harmonics do not change $\Pi$.

The representation of $w(\cdot)$ in Eq.~\eqref{eq:def_w} is a linear function of operators, then $w(\Omega_{\text{\rm Hom}})$ of the homogeneous QSV protocol is given by
\begin{eqnarray}\label{eq:pro_hom}
  w(\Omega_{\text{\rm Hom}})
  = (1-\nu)w(\openone)
  + \nu w(\psi)\,,
\end{eqnarray}
where $w(\openone)$ and $w(\psi)$ are the quasi-probability distributions of the identity and the projection on the target state, respectively. 
Hence, we have the corollary below for homogeneous QSV protocols.
\begin{corollary}
\label{cor:homQSV}
Considering a homogeneous QSV protocol $\Omega_{\text{\rm Hom}}$ for the target state \ket{\psi} as defined in Eq.~\eqref{eq:def_homoA}, it is local if the following two constraints are satisfied:
\begin{align}
  &\bullet{\bf Positivity}  &&\nu \leq \frac{1}{1-(2\pi)^n \min \{w(\psi)\}} \leq \frac{1}{2^{n-1}+1}\,,\label{eq:pos_hom}\\
  &\bullet{\bf Completeness} &&\cS(\psi) \leq 2^n\,.\label{eq:cop_hom}
\end{align}
\end{corollary}
\begin{proof}
From Eq.~\eqref{eq:def_w}, we directly get
\begin{eqnarray}
  w(\openone)={2^n}/{(4\pi)^n}
\end{eqnarray}
by considering the symmetry of the quasi-probability distribution.
Then, the positivity constraint of Eq.~\eqref{eq:pos_mea} in Theorem~\ref{thm:measure_c} is translated here for $\Omega_{\text{\rm Hom}}$, such that
\begin{eqnarray}
  \min_{\forall \vec{r}}\bigl\{w(\Omega_{\text{\rm Hom}})\bigr\}
  =(1-\nu){2^n}/{(4\pi)^n}+\nu\min_{\forall \vec{r}}\{w(\psi)\} \geq 0\,,
\end{eqnarray}
from which the inequality $\nu \!\leq\! 1/\bigl[1\!-\!(2\pi)^n \min \{w(\psi)\}\bigr]$ can be deduced.
Next, since the eigenvalues of $(\openone+3\vec{r}\cdot\vec{\sigma})$ are $4$ and $-2$, one obtains
\begin{eqnarray}
  \min \bigl\{w(\psi)\bigr\}=[4^{n-1}\times(-2)]/(4\pi)^n\,,
\end{eqnarray}
then the second inequality in Eq.~\eqref{eq:pos_hom} follows.

Finally, the completeness constraint of Eq.~\eqref{eq:cop_mea} in Theorem~\ref{thm:measure_c} is transformed here for $\Omega_{\text{\rm Hom}}$ as
\begin{eqnarray}
  \cS\bigl(\Omega_{\text{\rm Hom}}\bigr)=(1-\nu)2^n+\nu\cS(\psi)\leq 2^n\,.
\end{eqnarray}
Since ${0<\nu\leq 1}$, we have $\cS(\psi) \leq 2^n$, which is the completeness constraint for homogeneous QSV protocols.
\end{proof}

For the positivity constraint, we note that the first inequality in Eq.~\eqref{eq:pos_hom} gives the radius of the local ball;  
while the second one is obtained by considering all possible target states.
In other words, using the homogeneous QSV protocol, any target state can be verified with an efficiency no more than $O(2^n)$ as long as the completeness constraint is satisfied. 
Moreover, the upper bound of this complexity depends on the structure of the Hilbert space, such that more complex measurements like the multi-outcome POVMs will not improve the upper bound; see Appendix~B3 below.

\subsection{Multi-outcome measurements}\label{AppendixB3}
Here we modify Corollary~\ref{cor:homQSV} to the multi-outcome scenario by considering $s$-outcome rank-$1$ POVMs as $\{P'_{\vec{r}_1},\cdots,P'_{\vec{r}_s}\}$, where the measurements $P'_{\vec{r}_i}=\frac{1}{s}(\openone+\vec{r}_i\cdot\vec{\sigma})$ with $\sum_i P'_{\vec{r}_i} =\openone$. 
Then we have the following corollary for QSV.
\begin{corollary}\label{cor:multi}
  Considering a homogeneous QSV protocol $\Omega_\text{Hom}$ as defined in Eq.~\eqref{eq:def_homoA} for the target state $\ket{\psi}$, it is local under $s$-outcome rank-$1$ POVMs ({$s\geq 2$}) if the following two constraints are satisfied:
\begin{itemize}
  \item {\bf Positivity}
  \begin{equation}\label{eq:pos_homA}
    \nu^{(s)} \leq \frac{1}{1-(4\pi/s)^n \min \bigl\{w^{(s)}(\psi)\bigr\}} \leq \frac{1}{2^{n-1}+1}\,.
  \end{equation}
  \item {\bf Completeness}
  \begin{equation}
    \cS^{(s)}(\psi) \leq s^n\,.
  \end{equation}
\end{itemize}
\end{corollary}
\begin{proof}
For the POVMs $\{P'_{\vec{r}_1},\cdots,P'_{\vec{r}_s}\}$, we generalize Corollary~\ref{cor:homQSV} by following the remark of Theorem~\ref{thm:measure_c}. 
Being equivalent to Eq.~\eqref{eq:int}, one has
\begin{eqnarray}
    \Pi= \int \D\cV_1\cdots \D\cV_n w^{(s)}_{\vec{r}_1,\cdots,\vec{r}_n}(\Pi) P'_{\vec{r}_1}\otimes\cdots\otimes P'_{\vec{r}_n}\,,
\end{eqnarray}
where the quasi-probability distribution is $w^{(s)}(\Pi)=(s/2)^n w(\Pi)$, and we have 
\begin{eqnarray}
  w^{(s)}(\openone)={s^n}/{(4\pi)^n}\,.
\end{eqnarray}

The positivity constraint is generalized to
\begin{eqnarray}
  \min_{\forall \vec{r}}\bigl\{w^{(s)}(\Omega_{\text{\rm Hom}})\bigr\}
  =\bigl(1-\nu^{(s)}\bigr){s^n}/{(4\pi)^n}+\nu\min_{\forall \vec{r}}\bigl\{w^{(s)}(\psi)\bigr\} \geq 0\,,
\end{eqnarray}
and the inequality $\nu^{(s)} \leq 1/\bigl[1-(4\pi/s)^n \min \bigl\{w^{(s)}(\psi)\bigr\}\bigr]$ can be deduced.
Note that the eigenvalues of $(\openone+3\vec{r}\cdot\vec{\sigma})$ are $4$ and $-2$, one obtains
\begin{eqnarray}
  \min \bigl\{w^{(s)}(\psi)\bigr\}=\bigl[4^{n-1}\times(-2)\bigr]/(4\pi/s)^n=(-1/2)(1/\pi)^n(s/2)^n\,,
\end{eqnarray}
then the second inequality in Eq.~\eqref{eq:pos_homA} follows.

With the completeness constraint of Eq.~\eqref{eq:cop_mea} in Theorem~\ref{thm:measure_c}, we have 
\begin{eqnarray}
  \cS^{(s)}(\Omega_{\text{\rm Hom}})
  =\bigl(1-\nu^{(s)}\bigr)\int\bigl[s^n/(4\pi)^n]\D\cV+\nu^{(s)}\cS^{(s)}(\psi)
  \leq s^n\,.
\end{eqnarray}
Since $0<\nu^{(s)}\leq 1$, we have $\cS^{(s)}(\psi) \leq s^n$, which is the completeness constraint.
\end{proof}

\section{Proof of Corollary~\ref{cor:Pauli}}\label{AppendixC}
\begin{proof}
  For the identity $\openone$, all the coefficients of the Pauli representation are $0$ except for $c_{0\cdots 0}=2^n$, thus 
  \begin{eqnarray}
    p_i(\openone)=3^{-n}, ~\forall i\,.
  \end{eqnarray}
  Using the positivity constraint of  Eq.~\eqref{eq:pos_d} in Theorem~\ref{thm:measure_d} on the homogeneous protocol $\Omega_{\text{\rm Hom}}$, one has 
  \begin{eqnarray}
  \min_i\bigl\{p_i(\Omega_\text{Hom})\bigr\}=(1\!-\!\nu)\frac{1}{3^n}+\nu\min_{i}\bigl\{p_i(\psi)\bigr\} \geq 0\,,\quad
  \end{eqnarray}
  and the first inequality in Eq.~\eqref{eq:pos_homP} is deduced.
  
 Considering the transformation $T_\text{Pauli}$ with the form of Eq.~\eqref{eq:def_T}, from Eq.~\eqref{eq:Rep_Proj} we get
  \begin{eqnarray}
    \min_i \bigl\{p_{i_1 \cdots i_n}(\psi)\bigr\} 
    &=& \frac{1}{2^n}\Bigl(\sum_\alpha q_\alpha(\psi) \Bigr)_{i_1 \cdots i_n}\,,\\
    \bigl(q_\alpha(\psi)\bigr)_{i_1 \cdots i_n}
    &=& \left\{
        \begin{aligned}
            & 0\,, &&           (i_r\neq 1,2\text{ for }\alpha_r=1)
                    \text{ or }(i_r\neq 3,4\text{ for }\alpha_r=2)
                    \text{ or }(i_r\neq 5,6\text{ for }\alpha_r=3)\,,\\
            & \pm c_\alpha C_n^h \frac{1}{3^h}\,, &&\text{otherwise}\,,\\
        \end{aligned}
    \right.
  \end{eqnarray}
  where $h$ is the Hamming weight of the string $i_1 \cdots i_n$. 
  Since $-1\leq c_{\alpha_1\cdots\alpha_n}\leq 1$, the minimal value of $p_i(\psi)$ is obtained when $|c_\alpha|=1$ such that $(q_\alpha(\psi))_i < 0, \forall \alpha$, except for $c_{0\cdots 0}=1$.
  Thus we have  
  \begin{eqnarray}
    \min_{\forall \psi} \min_i \bigl\{p_i(\psi)\bigr\}
    = \frac{1}{2^n}
    \biggl(C_n^0 \frac{1}{3^n} - C_n^{1}\frac{1}{3^{n-1}} - \cdots - C_n^n \frac{1}{3^0}\biggr) 
    = \frac{1}{2^n}
    \biggl(\frac{2}{3^n}-\biggl(\frac{1}{3}+1\biggr)^n\biggr)
    =\frac{2-4^n}{2^n 3^n}\,,
  \end{eqnarray}
  and the second inequality in Eq.~\eqref{eq:pos_homP} is derived.
  
  Using the positivity constraint of  Eq.~\eqref{eq:cop_d} in Theorem~\ref{thm:measure_d} on the homogeneous protocol $\Omega_{\text{\rm Hom}}$, we have 
  \begin{eqnarray}
    S(\Omega_\text{Hom})=(1-\nu)\sum_{i\in\{1,2,3\}^{\otimes n}} \frac{1}{3^n} + \nu S(\psi)\leq 1\,.
  \end{eqnarray}
  Since $0<\nu\leq 1$, $S(\psi) \leq 1$ follows.
\end{proof}

\section{Additional details on the applications}\label{AppendixD}
Here we present more details on the protocol design for verifying Bell states, stabilizer states including GHZ states, and $W$ states. 
From numerical results to concrete realizations, we also show additional procedures on how to improve all the results.

\subsection{Bell states}
For the first example, we consider the Bell state $\ket{\Phi^{+}}=\frac1{\sqrt{2}}\bigl(\ket{00}+\ket{11}\bigr)$.
Its optimal QSV protocol \cite{Pallister.etal2018} is precisely homogeneous such that
\begin{eqnarray}
  \Omega_{\text{Bell}}
  =\Omega_{\text{Hom}}\bigl(\Phi^{+}\bigr)
  =\frac{1}{3}\bigl(P_{XX}^{+}+P_{YY}^{-}+P_{ZZ}^{+}\bigr)
  =\frac{1}{3}\openone+\frac{2}{3}\ket{\Phi^{+}}\bra{\Phi^{+}}\,,~\quad
\end{eqnarray}
where $X$, $Y$, and $Z$ are the Pauli operators, and the superscripts $+$ and $-$ indicate the projections onto the eigenspaces with eigenvalues $+1$ and $-1$ respectively. 
The verification efficiency is given by $1/\nu=3/2$. 

As shown by Corollary~\ref{cor:Pauli}, the constraints are directly related to the quasi-probability distribution $p\bigl(\Phi^{+}\bigr)$ of the target state.
With the transformation of Eq.~\eqref{eq:def_T}, it is
\begin{eqnarray}
  p\bigl(\Phi^{+}\bigr)=\frac{1}{36}\left[
  \begin{array}{rrrr}
    10  &-8 &-8 &10\\
    1   &1  &1  &1\\
    1   &1  &1  &1\\
    1   &1  &1  &1\\
    -8  &10 &10 &-8\\
    1   &1  &1  &1\\
    1   &1  &1  &1\\
    1   &1  &1  &1\\
    10  &-8 &-8 &10\\
  \end{array}
  \right]\!,
\end{eqnarray}
where $p_{ij}$ represents the coefficient for the local Pauli projection $P_{(-1)^{j_1}i_1}\otimes P_{(-1)^{j_2}i_2}$.
One notices that as an entangled state, some coefficients of the quasi-probability distribution $p\bigl(\Phi^{+}\bigr)$ are negative under local measurements.
However, it does satisfy ${S\bigl(\Phi^{+}\bigr)=1}$, meaning that the Bell state $\ket{\Phi^{+}}$ can be verified with the homogeneous protocol using local Pauli measurement settings only. 
Since the radius of the local ball is given by $\nu=1/(1-3^n \min\{p(\Phi^{+})\})=1/3$, then the quasi-probability distribution of the homogeneous protocol is 
\begin{eqnarray}\label{eq:p_OBell}
  p\bigl(\Omega_\text{Hom}\bigl(\Phi^{+}\bigr)\bigr)=\frac{1}{12}\left[
  \begin{array}{rrrr}
    2   &0  &0  &2\\
    1   &1  &1  &1\\
    1   &1  &1  &1\\
    1   &1  &1  &1\\
    0   &2  &2  &0\\
    1   &1  &1  &1\\
    1   &1  &1  &1\\
    1   &1  &1  &1\\
    2   &0  &0  &2\\
  \end{array}
  \right]\!,
\end{eqnarray}
where all the coefficients are nonnegative, and the verification efficiency is $1/\nu=3$. 

Furthermore, one notices that $\sum_{j}P_{(-1)^{j_1}i_1}\otimes P_{(-1)^{j_2}i_2}=\openone,~\forall i$. Then, some settings of the above protocol, such as the second row in Eq.~\eqref{eq:p_OBell}, indicate that the homogeneous protocol $\Omega_{\text{Hom}}$ contains null operations $\openone$, hence can be improved. 
We add a revision process to delete the additional null operation $a\openone$ with an appropriate scaling factor $a$, thus in general we have 
\begin{eqnarray}\label{eq:revision}
  \Omega_\text{Hom}^{\prime}\bigl(\Phi^{+}\bigr)
  =\frac{\Omega_\text{Hom}\bigl(\Phi^{+}\bigr)-a\openone}{1-a}
  =\frac{1-\nu-a}{1-a}\openone + \frac{\nu}{1-a}\ket{\psi}\bra{\psi}
  =\bigl(1-\nu'\bigr)\openone + \nu' \ket{\psi}\bra{\psi}\,,
  \quad \text{with}\quad \nu'=\frac{\nu}{1-a}\,.
\end{eqnarray}
Hence, with the revision process, the improved protocol $\Omega_\text{Hom}^{\prime}$ for the Bell state is exactly the same as the optimal one, for which we have $a=1/2$, and the quasi-probability distribution is 
\begin{eqnarray}
  p\bigl(\Omega_\text{Hom}^{\prime}\bigl(\Phi^{+}\bigr)\bigr)=\frac{1}{3}\left[
  \begin{array}{rrrr}
    1   &0  &0  &1\\
    0   &0  &0  &0\\
    0   &0  &0  &0\\
    0   &0  &0  &0\\
    0   &1  &1  &0\\
    0   &0  &0  &0\\
    0   &0  &0  &0\\
    0   &0  &0  &0\\
    1   &0  &0  &1
  \end{array}
  \right]\!.
\end{eqnarray}

\subsection{Stabilizer states including GHZ states}
A standard approach for verifying stabilizer states \cite{Pallister.etal2018} is to construct the protocol by using all the stabilizer generators with equal weight, then the efficiency is given by $1/\nu=n$; or by using the full set of $2^n-1$ linearly independent stabilizers with equal weight and the efficiency is $1/\nu=(2^n-1)/2^{n-1}$.
For GHZ states in particular, an optimal homogeneous verification protocol can be constructed as 
\begin{eqnarray}
  \Omega_\text{GHZ} = \frac{1}{3}\biggl(P_0+\frac{1}{2^{n-2}}\sum_{\mathcal{Y}}P_\mathcal{Y}\biggr)=\frac{1}{3}\openone+\frac{2}{3}\ket{\text{GHZ}}\bra{\text{GHZ}}\,,
\end{eqnarray}
where $P_0$ represents the Pauli-$Z$ projection on each party with the same outcomes, and $P_\mathcal{Y}$ denotes all the possible combinations of the local Pauli-$X$ and Pauli-$Y$ on each party; see Ref.~\cite{Li.etal2020b} for more detailed descriptions.

Considering three-qubit GHZ state $\ket{\text{GHZ}_3}=\frac1{\sqrt{2}}(\ket{000}+\ket{111})$, the efficiency with stabilizers is $1/\nu=7/4$, and the optimal efficiency is $1/\nu=3/2$. 
Our method shows that it can be verified homogeneously with local Pauli projections as $S\bigl(\text{GHZ}_3\bigr)=1$. 
The efficiency is given by $1/\nu=17/4$.
In addition, for the $n$-qubit GHZ state we have the following corollary.
\begin{corollary}
    For the $n$-qubit GHZ state, there exists a homogeneous QSV protocol using Pauli projections with the transformation matrix as defined in Eq.~\eqref{eq:def_T}, of which we can achieve the efficiency 
    \begin{eqnarray}
        \frac{1}{\nu} = \frac{3^n+2^n-1}{2^n}\,.
    \end{eqnarray}
\end{corollary}
\begin{proof}
    Using the three-qubit GHZ state as an example, we can find the coefficients of the Pauli representation with the following property
    \begin{eqnarray}
        c_{000} &=& \tr \bra{\text{GHZ}_3}\sigma_0 \otimes \sigma_0 \otimes \sigma_0 \ket{\text{GHZ}_3} = 1\,,\nonumber\\
        c_{001} &=& \tr \bra{\text{GHZ}_3}\sigma_0 \otimes \sigma_0 \otimes \sigma_1 \ket{\text{GHZ}_3} = 0\,,\nonumber\\
        c_{002} &=& \tr \bra{\text{GHZ}_3}\sigma_0 \otimes \sigma_0 \otimes \sigma_2 \ket{\text{GHZ}_3} = 0\,,\nonumber\\
        c_{003} &=& \tr \bra{\text{GHZ}_3}\sigma_0 \otimes \sigma_0 \otimes \sigma_3 \ket{\text{GHZ}_3} = 0\,,\nonumber\\
        \cdots && \nonumber\\
        c_{333} &=& \tr \bra{\text{GHZ}_3}\sigma_3 \otimes \sigma_3 \otimes \sigma_3 \ket{\text{GHZ}_3} = 0\,.
    \end{eqnarray}
    Since $\sigma_2=\I\sigma_1 \sigma_3$, for the $n$-qubit GHZ state, besides the coefficient $c_{\{0\}^{\otimes n}}$ of the identity operator, only the coefficients $c_{\{1\}^{\otimes n}}$ and $c_{\{0,3\}^{\otimes n}}$ with even Hamming weight are $1$ and others are all zero. Thus, with Eq.~\eqref{eq:Rep_Proj} and the transformation Eq.~\eqref{eq:def_T}, for $i_1 \cdots i_n \in \{1,2,3,4,5,6\}^{\otimes n}$, we have 
    \begin{eqnarray}
        p_{i_1 \cdots i_n}({\text{GHZ}_3}) = \frac{1}{2^n}
        \left\{
        \begin{aligned}
            & \frac{1}{3^n}\pm 1\,, &&  i_1\cdots i_n \in \{1,2\}^{\otimes n}\text{ or }i_1\cdots i_n \in \{5,6\}^{\otimes n}\,,\\
            & \frac{1}{3^n}\,, &&\text{otherwise}\,.\\
        \end{aligned}
    \right.
    \end{eqnarray}
    Combining Corollary~\ref{cor:Pauli}, we have
    \begin{eqnarray}
        \frac{1}{\nu} = 1 - 3^n \min\{p_i({\text{GHZ}_3})\} = 1 - 3^n \frac{1}{2^n}\left(\frac{1}{3^n}-1\right) = \frac{3^n+2^n-1}{2^n}\,.
    \end{eqnarray}
\end{proof}

Furthermore, this efficiency can be improved to $1/\nu=5/3$ with an additional revision process.
This is better than that in Ref.~\cite{Pallister.etal2018}, and slightly worse than the optimal one.
The quasi-probability distribution is 
\begin{eqnarray}
  p\bigl(\Omega_\text{Hom}\bigl(\text{GHZ}_3\bigr)\bigr)=\frac{1}{20}\left[
  \begin{array}{rrrrrrrr}
    3 &0  &0  &3  &0  &3  &3  &0  \\
    0 &0  &0  &0  &0  &0  &0  &0  \\
    0 &0  &0  &0  &0  &0  &0  &0  \\
    0 &0  &0  &0  &0  &0  &0  &0  \\
    0 &3  &3  &0  &3  &0  &0  &3  \\
    0 &0  &0  &0  &0  &0  &0  &0  \\
    0 &0  &0  &0  &0  &0  &0  &0  \\
    0 &0  &0  &0  &0  &0  &0  &0  \\
    1 &0  &0  &1  &1  &0  &0  &1  \\
    0 &0  &0  &0  &0  &0  &0  &0  \\
    0 &3  &3  &0  &3  &0  &0  &3  \\
    0 &0  &0  &0  &0  &0  &0  &0  \\
    0 &3  &3  &0  &3  &0  &0  &3  \\
    0 &0  &0  &0  &0  &0  &0  &0  \\
    0 &0  &0  &0  &0  &0  &0  &0  \\
    0 &0  &0  &0  &0  &0  &0  &0  \\
    0 &0  &0  &0  &0  &0  &0  &0  \\
    1 &0  &0  &1  &1  &0  &0  &1  \\
    0 &0  &0  &0  &0  &0  &0  &0  \\
    0 &0  &0  &0  &0  &0  &0  &0  \\
    1 &0  &1  &0  &0  &1  &0  &1  \\
    0 &0  &0  &0  &0  &0  &0  &0  \\
    0 &0  &0  &0  &0  &0  &0  &0  \\
    1 &0  &1  &0  &0  &1  &0  &1  \\
    1 &1  &0  &0  &0  &0  &1  &1  \\
    1 &1  &0  &0  &0  &0  &1  &1  \\
    2 &0  &0  &0  &0  &0  &0  &2  
  \end{array}
  \right]\!.
\end{eqnarray}

Furthermore, with the fact that all stabilizer states can be verified by QSV protocols constructed with their stabilizers which are in the Pauli group \cite{Pallister.etal2018}, our method is able to give the local homogeneous QSV protocols for all stabilizer states with the quasi-probability distribution based on the Pauli representation. 
This can be shown numerically such that we have checked all the graph states up to five qubits (which are equivalent to stabilizer states).

\subsection{$W$ states}
$W$ states (or more generally, Dicke states) have been efficiently verified in our previous work \cite{Liu.etal2019b} using only local Pauli-$Z$ and Pauli-$X$ measurements. 
The efficiency is $1/\nu=n-1$ for $n\geq 4$ ($1/\nu=3$ for $n=3$) with adaptive measurements and is worsened by a factor of 2 with nonadaptive measurements. 
In addition, Li. \textit{et al.} \cite{Li.etal2020a} proposed a nearly optimal protocol with the efficiency of $1/\nu=8/5$, which is also homogeneous. 
However, besides the Pauli-$X$ and Pauli-$Z$ measurements, their protocol requires an additional projection on $(2\ket{0}\pm\ket{1})/\sqrt{5}$ as well as certain symmetrization procedures.

Consider the three-qubit $W$ state $\ket{W_3}=\frac1{\sqrt{3}}(\ket{001}+\ket{010}+\ket{100})$. Unfortunately, one has $S\bigl(W_3\bigr)=1.40(7)$ by using our method, which violates the completeness constraint. 
In turn, the revision process cannot make the constraint be satisfied either.
The revised quasi-probability distribution is 
\begin{eqnarray}
  p\bigl(\Omega_\text{Hom}(W_3)\bigr)=\frac{1}{222}\left[
  \begin{array}{rrrrrrrr}
    12  &0  &0  &0  &0  &0  &0  &12\\
    6   &6  &0  &0  &0  &0  &6  &6\\
    24  &5  &0  &17 &0  &17 &24 &5\\
    6   &0  &6  &0  &0  &6  &0  &6\\
    6   &0  &0  &6  &6  &0  &0  &6\\
    1   &0  &1  &0  &1  &0  &1  &0\\
    24  &0  &5  &17 &0  &24 &17 &5\\
    1   &1  &0  &0  &1  &1  &0  &0\\
    2   &4  &4  &0  &2  &4  &4  &0\\
    6   &0  &0  &6  &6  &0  &0  &6\\
    6   &0  &6  &0  &0  &6  &0  &6\\
    1   &0  &1  &0  &1  &0  &1  &0\\
    6   &6  &0  &0  &0  &0  &6  &6\\
    12  &0  &0  &0  &0  &0  &0  &12\\
    24  &5  &0  &17 &0  &17 &24 &5\\
    1   &1  &0  &0  &1  &1  &0  &0\\
    24  &0  &5  &17 &0  &24 &17 &5\\
    2   &4  &4  &0  &2  &4  &4  &0\\
    24  &0  &0  &24 &5  &17 &17 &5\\
    1   &1  &1  &1  &0  &0  &0  &0\\
    2   &4  &2  &4  &4  &0  &4  &0\\
    1   &1  &1  &1  &0  &0  &0  &0\\
    24  &0  &0  &24 &5  &17 &17 &5\\
    2   &4  &2  &4  &4  &0  &4  &0\\
    2   &2  &4  &4  &4  &4  &0  &0\\
    2   &2  &4  &4  &4  &4  &0  &0\\
    0   &32 &32 &4  &32 &4  &4  &24
  \end{array}
  \right]\!,
\end{eqnarray}
with $S\bigl(\Omega_\text{Hom}(W_3)\bigr)=1.19(8)>1$. 

On the other hand, since we have the quasi-probability distribution, further analysis is still meaningful.
We find that not only the operator $\Omega_\text{Hom}(W_3)$ does not satisfy the constraints, but some of the measurement settings do not fulfill $\tr\bigl(\Omega_i\ket{W_3}\bra{W_3}\bigr)=1$. 
Then, we pick out the settings which do satisfy $\tr\bigl(\Omega_i\ket{W_3}\bra{W_3}\bigr)=1$, and use them to construct a verification protocol with a uniform probability distribution. 
Hence, we have 
\begin{eqnarray}
  \Omega'(W_3)=\frac{1}{13}\left[
  \begin{array}{rrrrrrrr}
    9   & 0   & 0   &0   & 0   &0    &0    &0\\
    0   &11   & 1   &0   & 1   &0    &0    &0\\
    0   & 1   &11   &0   & 1   &0    &0    &0\\
    0   & 0   & 0   &9   & 0   &0    &0    &0\\
    0   & 1   & 1   &0   &11   &0    &0    &0\\
    0   & 0   & 0   &0   & 0   &9    &0    &0\\
    0   & 0   & 0   &0   & 0   &0    &9    &0\\
    0   & 0   & 0   &0   & 0   &0    &0    &6
  \end{array}
  \right]\!,
\end{eqnarray}
which is not homogeneous.
However, it is a valid protocol for verifying the three-qubit $W$ state with the efficiency given by $1/\nu=13/3$, which is better than that of the previous inhomogeneous protocol with Pauli projections \cite{Liu.etal2019b}. 

\subsection{Different choices of the transformation $T_\text{Pauli}$}
With the specific choice of the transformation between the Pauli operators $\{\sigma_0,\sigma_1,\sigma_2,\sigma_3\}$ and the Pauli projections $\{P_i^0,P_i^1\}_{i=1}^3$ as in Eq.~\eqref{eq:def_T}
\begin{eqnarray}
  T_\text{Pauli}:=
  \begin{bmatrix}
    \vec{t}_0\\
    \vec{t}_1\\
    \vec{t}_2\\
    \vec{t}_3
  \end{bmatrix}
  =
  \begin{bmatrix}
    1/3   &1/3    &1/3    &1/3    &1/3    &1/3  \\
    1     &-1     &0      &0      &0      &0    \\
    0     &0      &1      &-1     &0      &0    \\
    0     &0      &0      &0      &1      &-1   \\
  \end{bmatrix}
  \!,
\end{eqnarray}
we get Corollary~\ref{cor:Pauli} as well as all the results of the previous applications. 
As mentioned in the main text, the choice of $\vec{t}_0$ is arbitrary such that 
\begin{eqnarray}
  \openone = P_1^0+P_1^1 = P_2^0+P_2^1 = P_3^0+P_3^1 = \sum_i\alpha_i (P_i^0 + P_i^1)\,,
\end{eqnarray}
where $\sum_i \alpha_i =1$.
Thus, in general, one has 
\begin{eqnarray}
  \vec{t}_0 = [\alpha_1~\alpha_1~\alpha_2~\alpha_2~\alpha_3~\alpha_3]\,.
\end{eqnarray}
Obviously, Corollary~\ref{cor:Pauli} is not valid anymore with a different transformation. 
Reconsidering Theorem~\ref{thm:measure_d}, we have 
\begin{eqnarray}
  \min_i\bigl\{p_i(\Omega_\text{Hom})\bigr\}=\min_i\bigl\{(1\!-\!\nu)p_i(\openone)+\nu p_i(\psi)\bigr\} \geq 0\,,\quad
\end{eqnarray}
with $p_{i_1 \cdots i_n}(\openone)=\bigl(\vec{t}_0^{\otimes n}\bigr)_{i_1 \cdots i_n}$. Under such a circumstance, it is difficult to give a general bound for $\nu$.
However, one finds that, the revision process in Eq.~\eqref{eq:revision} does not require $a>0$. 
Hence, we get
\begin{eqnarray}
  \Omega_\text{Hom}(\psi)
  =\frac{\ket{\psi}\bra{\psi}-a\openone}{1-a}
  =\frac{-a}{1-a}\openone + \frac{1}{1-a}\ket{\psi}\bra{\psi}
  =(1-\nu)\openone + \nu \ket{\psi}\bra{\psi}\,,
  \quad \text{with}\quad \nu=\frac{1}{1-a}\,.
\end{eqnarray}

Taking the three-qubit GHZ state $\ket{\text{GHZ}_3}=\frac1{\sqrt{2}}\bigl(\ket{000}+\ket{111}\bigr)$ as an example, with the transformation $\vec{t}_0=[0~~0~~0~~0~~1~~1]$, our method along with the revision process is able to give a local homogeneous protocol with a better efficiency of $1/\nu=3/2$, which happens to be the optimal one \cite{Li.etal2020b}. 

More importantly, for the three-qubit $W$ state $\ket{W_3}=\frac1{\sqrt{3}}\bigl(\ket{001}+\ket{010}+\ket{100}\bigr)$, with the transformation $\vec{t}_0=[0~~0~~0~~0~~1~~1]$, we have the quasi-probability distribution
\begin{eqnarray}
  p\bigl(\Omega_\text{Hom}(W_3)\bigr)=\frac{1}{12}\left[
  \begin{array}{rrrrrrrr}
    0  &0  &0  &0  &0  &0  &0  &0\\
    0  &0  &0  &0  &0  &0  &0  &0\\
    2  &1  &0  &1  &0  &1  &2  &1\\
    0  &0  &0  &0  &0  &0  &0  &0\\
    0  &0  &0  &0  &0  &0  &0  &0\\
    0  &0  &0  &0  &0  &0  &0  &0\\
    2  &0  &1  &1  &0  &2  &1  &1\\
    0  &0  &0  &0  &0  &0  &0  &0\\
    0  &0  &0  &0  &0  &0  &0  &0\\
    0  &0  &0  &0  &0  &0  &0  &0\\
    0  &0  &0  &0  &0  &0  &0  &0\\
    0  &0  &0  &0  &0  &0  &0  &0\\
    0  &0  &0  &0  &0  &0  &0  &0\\
    0  &0  &0  &0  &0  &0  &0  &0\\
    2  &1  &0  &1  &0  &1  &2  &1\\
    0  &0  &0  &0  &0  &0  &0  &0\\
    2  &0  &1  &1  &0  &2  &1  &1\\
    0  &0  &0  &0  &0  &0  &0  &0\\
    2  &0  &0  &2  &1  &1  &1  &1\\
    0  &0  &0  &0  &0  &0  &0  &0\\
    0  &0  &0  &0  &0  &0  &0  &0\\
    0  &0  &0  &0  &0  &0  &0  &0\\
    2  &0  &0  &2  &1  &1  &1  &1\\
    0  &0  &0  &0  &0  &0  &0  &0\\
    0  &0  &0  &0  &0  &0  &0  &0\\
    0  &0  &0  &0  &0  &0  &0  &0\\
    0  &2  &2  &0  &2  &0  &0  &0
  \end{array}
  \right]\!,
\end{eqnarray}
with $S\bigl(\Omega_\text{Hom}(W_3)\bigr)=7/6>1$. 
However, considering adaptive measurements as in Theorem~\ref{thm:measure_a}, the constraint can be satisfied such that $S\bigl(\Omega_\text{Hom}(W_3)\bigr)=7/6<4$. 
Therefore, we are able to get a homogeneous QSV protocol for $\ket{W_3}$ using local Pauli projections as
\begin{eqnarray}
    \Omega_\text{Hom}(W_3)=\sum_k \frac{1}{3} \cP_k\biggl\{
    P_Z^+\Bigl[\frac{1}{2}P_{XX}^+ + \frac{1}{2}P_{YY}^+\Bigr]
    +P_Z^- \Bigl[\frac{1}{2}\openone\openone+\frac{1}{2}P_{Z}^+P_{Z}^+\Bigr]\biggr\}\,,
\end{eqnarray}
where $\cP_k\{\cdot\}(k=1,2,3)$ is the permutation of qubits. This gives us a better verification protocol with an efficiency of $1/\nu =2$. It is better than the previous protocol using adaptive local Pauli projections, which is also not homogeneous \cite{Liu.etal2019b}. 
Compared with the nearly optimal homogeneous protocol with the efficiency of $1/\nu=8/5$ which demands an additional projection on $(2\ket{0}\pm\ket{1})/\sqrt{5}$ as well as certain symmetrization procedures \cite{Li.etal2020a}, our new homogeneous protocol is much easier to realize.


\begin{thebibliography}{45}%
  \makeatletter
  \providecommand \@ifxundefined [1]{%
   \@ifx{#1\undefined}
  }%
  \providecommand \@ifnum [1]{%
   \ifnum #1\expandafter \@firstoftwo
   \else \expandafter \@secondoftwo
   \fi
  }%
  \providecommand \@ifx [1]{%
   \ifx #1\expandafter \@firstoftwo
   \else \expandafter \@secondoftwo
   \fi
  }%
  \providecommand \natexlab [1]{#1}%
  \providecommand \enquote  [1]{``#1''}%
  \providecommand \bibnamefont  [1]{#1}%
  \providecommand \bibfnamefont [1]{#1}%
  \providecommand \citenamefont [1]{#1}%
  \providecommand \href@noop [0]{\@secondoftwo}%
  \providecommand \href [0]{\begingroup \@sanitize@url \@href}%
  \providecommand \@href[1]{\@@startlink{#1}\@@href}%
  \providecommand \@@href[1]{\endgroup#1\@@endlink}%
  \providecommand \@sanitize@url [0]{\catcode `\\12\catcode `\$12\catcode
    `\&12\catcode `\#12\catcode `\^12\catcode `\_12\catcode `\%12\relax}%
  \providecommand \@@startlink[1]{}%
  \providecommand \@@endlink[0]{}%
  \providecommand \url  [0]{\begingroup\@sanitize@url \@url }%
  \providecommand \@url [1]{\endgroup\@href {#1}{\urlprefix }}%
  \providecommand \urlprefix  [0]{URL }%
  \providecommand \Eprint [0]{\href }%
  \providecommand \doibase [0]{https://doi.org/}%
  \providecommand \selectlanguage [0]{\@gobble}%
  \providecommand \bibinfo  [0]{\@secondoftwo}%
  \providecommand \bibfield  [0]{\@secondoftwo}%
  \providecommand \translation [1]{[#1]}%
  \providecommand \BibitemOpen [0]{}%
  \providecommand \bibitemStop [0]{}%
  \providecommand \bibitemNoStop [0]{.\EOS\space}%
  \providecommand \EOS [0]{\spacefactor3000\relax}%
  \providecommand \BibitemShut  [1]{\csname bibitem#1\endcsname}%
  \let\auto@bib@innerbib\@empty
  \bibitem [{\citenamefont {Arute}\ \emph {et~al.}(2019)\citenamefont {Arute},
    \citenamefont {Arya}, \citenamefont {Babbush}, \citenamefont {Bacon},
    \citenamefont {Bardin}, \citenamefont {Barends}, \citenamefont {Biswas},
    \citenamefont {Boixo}, \citenamefont {Brandao}, \citenamefont {Buell} \emph
    {et~al.}}]{Google2019Sycamore}%
    \BibitemOpen
    \bibfield  {author} {\bibinfo {author} {\bibfnamefont {F.}~\bibnamefont
    {Arute}}, \bibinfo {author} {\bibfnamefont {K.}~\bibnamefont {Arya}},
    \bibinfo {author} {\bibfnamefont {R.}~\bibnamefont {Babbush}}, \bibinfo
    {author} {\bibfnamefont {D.}~\bibnamefont {Bacon}}, \bibinfo {author}
    {\bibfnamefont {J.~C.}\ \bibnamefont {Bardin}}, \bibinfo {author}
    {\bibfnamefont {R.}~\bibnamefont {Barends}}, \bibinfo {author} {\bibfnamefont
    {R.}~\bibnamefont {Biswas}}, \bibinfo {author} {\bibfnamefont
    {S.}~\bibnamefont {Boixo}}, \bibinfo {author} {\bibfnamefont {F.~G.}\
    \bibnamefont {Brandao}}, \bibinfo {author} {\bibfnamefont {D.~A.}\
    \bibnamefont {Buell}}, \emph {et~al.},\ }\bibfield  {title} {\bibinfo {title}
    {Quantum supremacy using a programmable superconducting processor},\ }\href
    {https://doi.org/10.1038/s41586-019-1666-5} {\bibfield  {journal} {\bibinfo
    {journal} {Nature}\ }\textbf {\bibinfo {volume} {574}},\ \bibinfo {pages}
    {505} (\bibinfo {year} {2019})}\BibitemShut {NoStop}%
  \bibitem [{\citenamefont {Zhong}\ \emph {et~al.}(2020)\citenamefont {Zhong},
    \citenamefont {Wang}, \citenamefont {Deng}, \citenamefont {Chen},
    \citenamefont {Peng}, \citenamefont {Luo}, \citenamefont {Qin}, \citenamefont
    {Wu}, \citenamefont {Ding}, \citenamefont {Hu} \emph
    {et~al.}}]{USTC2020Jiuzhang}%
    \BibitemOpen
    \bibfield  {author} {\bibinfo {author} {\bibfnamefont {H.-S.}\ \bibnamefont
    {Zhong}}, \bibinfo {author} {\bibfnamefont {H.}~\bibnamefont {Wang}},
    \bibinfo {author} {\bibfnamefont {Y.-H.}\ \bibnamefont {Deng}}, \bibinfo
    {author} {\bibfnamefont {M.-C.}\ \bibnamefont {Chen}}, \bibinfo {author}
    {\bibfnamefont {L.-C.}\ \bibnamefont {Peng}}, \bibinfo {author}
    {\bibfnamefont {Y.-H.}\ \bibnamefont {Luo}}, \bibinfo {author} {\bibfnamefont
    {J.}~\bibnamefont {Qin}}, \bibinfo {author} {\bibfnamefont {D.}~\bibnamefont
    {Wu}}, \bibinfo {author} {\bibfnamefont {X.}~\bibnamefont {Ding}}, \bibinfo
    {author} {\bibfnamefont {Y.}~\bibnamefont {Hu}}, \emph {et~al.},\ }\bibfield
    {title} {\bibinfo {title} {Quantum computational advantage using photons},\
    }\href {https://doi.org/10.1126/science.abe8770} {\bibfield  {journal}
    {\bibinfo  {journal} {Science}\ }\textbf {\bibinfo {volume} {370}},\ \bibinfo
    {pages} {1460} (\bibinfo {year} {2020})}\BibitemShut {NoStop}%
  \bibitem [{\citenamefont {Wu}\ \emph {et~al.}(2021)\citenamefont {Wu},
    \citenamefont {Bao}, \citenamefont {Cao}, \citenamefont {Chen}, \citenamefont
    {Chen}, \citenamefont {Chen}, \citenamefont {Chung}, \citenamefont {Deng},
    \citenamefont {Du}, \citenamefont {Fan} \emph {et~al.}}]{USTC2021zuchongzhi}%
    \BibitemOpen
    \bibfield  {author} {\bibinfo {author} {\bibfnamefont {Y.}~\bibnamefont
    {Wu}}, \bibinfo {author} {\bibfnamefont {W.-S.}\ \bibnamefont {Bao}},
    \bibinfo {author} {\bibfnamefont {S.}~\bibnamefont {Cao}}, \bibinfo {author}
    {\bibfnamefont {F.}~\bibnamefont {Chen}}, \bibinfo {author} {\bibfnamefont
    {M.-C.}\ \bibnamefont {Chen}}, \bibinfo {author} {\bibfnamefont
    {X.}~\bibnamefont {Chen}}, \bibinfo {author} {\bibfnamefont {T.-H.}\
    \bibnamefont {Chung}}, \bibinfo {author} {\bibfnamefont {H.}~\bibnamefont
    {Deng}}, \bibinfo {author} {\bibfnamefont {Y.}~\bibnamefont {Du}}, \bibinfo
    {author} {\bibfnamefont {D.}~\bibnamefont {Fan}}, \emph {et~al.},\ }\bibfield
     {title} {\bibinfo {title} {Strong quantum computational advantage using a
    superconducting quantum processor},\ }\href
    {https://doi.org/10.1103/PhysRevLett.127.180501} {\bibfield  {journal}
    {\bibinfo  {journal} {Phys. Rev. Lett.}\ }\textbf {\bibinfo {volume} {127}},\
    \bibinfo {pages} {180501} (\bibinfo {year} {2021})}\BibitemShut {NoStop}%
  \bibitem [{\citenamefont {Cozzolino}\ \emph {et~al.}(2019)\citenamefont
    {Cozzolino}, \citenamefont {Da~Lio}, \citenamefont {Bacco},\ and\
    \citenamefont {Oxenl{\o}we}}]{Cozzolino.etal2019}%
    \BibitemOpen
    \bibfield  {author} {\bibinfo {author} {\bibfnamefont {D.}~\bibnamefont
    {Cozzolino}}, \bibinfo {author} {\bibfnamefont {B.}~\bibnamefont {Da~Lio}},
    \bibinfo {author} {\bibfnamefont {D.}~\bibnamefont {Bacco}},\ and\ \bibinfo
    {author} {\bibfnamefont {L.~K.}\ \bibnamefont {Oxenl{\o}we}},\ }\bibfield
    {title} {\bibinfo {title} {High-dimensional quantum communication: benefits,
    progress, and future challenges},\ }\href
    {https://doi.org/10.1002/qute.201900038} {\bibfield  {journal} {\bibinfo
    {journal} {Adv. Quantum Technol.}\ }\textbf {\bibinfo {volume} {2}},\
    \bibinfo {pages} {1900038} (\bibinfo {year} {2019})}\BibitemShut {NoStop}%
  \bibitem [{\citenamefont {Bhaskar}\ \emph {et~al.}(2020)\citenamefont
    {Bhaskar}, \citenamefont {Riedinger}, \citenamefont {Machielse},
    \citenamefont {Levonian}, \citenamefont {Nguyen}, \citenamefont {Knall},
    \citenamefont {Park}, \citenamefont {Englund}, \citenamefont {Lon{\v{c}}ar},
    \citenamefont {Sukachev} \emph {et~al.}}]{Bhaskar.etal2020}%
    \BibitemOpen
    \bibfield  {author} {\bibinfo {author} {\bibfnamefont {M.~K.}\ \bibnamefont
    {Bhaskar}}, \bibinfo {author} {\bibfnamefont {R.}~\bibnamefont {Riedinger}},
    \bibinfo {author} {\bibfnamefont {B.}~\bibnamefont {Machielse}}, \bibinfo
    {author} {\bibfnamefont {D.~S.}\ \bibnamefont {Levonian}}, \bibinfo {author}
    {\bibfnamefont {C.~T.}\ \bibnamefont {Nguyen}}, \bibinfo {author}
    {\bibfnamefont {E.~N.}\ \bibnamefont {Knall}}, \bibinfo {author}
    {\bibfnamefont {H.}~\bibnamefont {Park}}, \bibinfo {author} {\bibfnamefont
    {D.}~\bibnamefont {Englund}}, \bibinfo {author} {\bibfnamefont
    {M.}~\bibnamefont {Lon{\v{c}}ar}}, \bibinfo {author} {\bibfnamefont {D.~D.}\
    \bibnamefont {Sukachev}}, \emph {et~al.},\ }\bibfield  {title} {\bibinfo
    {title} {Experimental demonstration of memory-enhanced quantum
    communication},\ }\href {https://doi.org/10.1038/s41586-020-2103-5}
    {\bibfield  {journal} {\bibinfo  {journal} {Nature}\ }\textbf {\bibinfo
    {volume} {580}},\ \bibinfo {pages} {60} (\bibinfo {year} {2020})}\BibitemShut
    {NoStop}%
  \bibitem [{\citenamefont {Pezz\`e}\ \emph {et~al.}(2018)\citenamefont
    {Pezz\`e}, \citenamefont {Smerzi}, \citenamefont {Oberthaler}, \citenamefont
    {Schmied},\ and\ \citenamefont {Treutlein}}]{Pezze.etal2018}%
    \BibitemOpen
    \bibfield  {author} {\bibinfo {author} {\bibfnamefont {L.}~\bibnamefont
    {Pezz\`e}}, \bibinfo {author} {\bibfnamefont {A.}~\bibnamefont {Smerzi}},
    \bibinfo {author} {\bibfnamefont {M.~K.}\ \bibnamefont {Oberthaler}},
    \bibinfo {author} {\bibfnamefont {R.}~\bibnamefont {Schmied}},\ and\ \bibinfo
    {author} {\bibfnamefont {P.}~\bibnamefont {Treutlein}},\ }\bibfield  {title}
    {\bibinfo {title} {Quantum metrology with nonclassical states of atomic
    ensembles},\ }\href {https://doi.org/10.1103/RevModPhys.90.035005} {\bibfield
     {journal} {\bibinfo  {journal} {Rev. Mod. Phys.}\ }\textbf {\bibinfo
    {volume} {90}},\ \bibinfo {pages} {035005} (\bibinfo {year}
    {2018})}\BibitemShut {NoStop}%
  \bibitem [{\citenamefont {Polino}\ \emph {et~al.}(2020)\citenamefont {Polino},
    \citenamefont {Valeri}, \citenamefont {Spagnolo},\ and\ \citenamefont
    {Sciarrino}}]{Polino.etal2020}%
    \BibitemOpen
    \bibfield  {author} {\bibinfo {author} {\bibfnamefont {E.}~\bibnamefont
    {Polino}}, \bibinfo {author} {\bibfnamefont {M.}~\bibnamefont {Valeri}},
    \bibinfo {author} {\bibfnamefont {N.}~\bibnamefont {Spagnolo}},\ and\
    \bibinfo {author} {\bibfnamefont {F.}~\bibnamefont {Sciarrino}},\ }\bibfield
    {title} {\bibinfo {title} {Photonic quantum metrology},\ }\href
    {https://doi.org/10.1116/5.0007577} {\bibfield  {journal} {\bibinfo
    {journal} {AVS Quantum Sci.}\ }\textbf {\bibinfo {volume} {2}},\ \bibinfo
    {pages} {024703} (\bibinfo {year} {2020})}\BibitemShut {NoStop}%
  \bibitem [{\citenamefont {Paris}\ and\ \citenamefont
    {\v{R}eh\'a\v{c}ek}(2004)}]{QSE2004}%
    \BibitemOpen
    \bibinfo {editor} {\bibfnamefont {M.}~\bibnamefont {Paris}}\ and\ \bibinfo
    {editor} {\bibfnamefont {J.}~\bibnamefont {\v{R}eh\'a\v{c}ek}},\ eds.,\
    \href@noop {} {\emph {\bibinfo {title} {Quantum State Estimation}}},\
    \bibinfo {series} {Lecture Notes in Physics}, Vol.\ \bibinfo {volume} {649}\
    (\bibinfo  {publisher} {Springer-Verlag Berlin Heidelberg},\ \bibinfo {year}
    {2004})\BibitemShut {NoStop}%
  \bibitem [{\citenamefont {H\"{a}ffner}\ \emph {et~al.}(2005)\citenamefont
    {H\"{a}ffner}, \citenamefont {H\"{a}nsel}, \citenamefont {Roos},
    \citenamefont {Benhelm}, \citenamefont {{Chek-al-kar}}, \citenamefont
    {Chwalla}, \citenamefont {K\"{o}rber}, \citenamefont {Rapol}, \citenamefont
    {Riebe}, \citenamefont {Schmidt}, \citenamefont {Becher}, \citenamefont
    {G\"{u}hne}, \citenamefont {D\"{u}r},\ and\ \citenamefont
    {Blatt}}]{Haffner.etal2005}%
    \BibitemOpen
    \bibfield  {author} {\bibinfo {author} {\bibfnamefont {H.}~\bibnamefont
    {H\"{a}ffner}}, \bibinfo {author} {\bibfnamefont {W.}~\bibnamefont
    {H\"{a}nsel}}, \bibinfo {author} {\bibfnamefont {C.~F.}\ \bibnamefont
    {Roos}}, \bibinfo {author} {\bibfnamefont {J.}~\bibnamefont {Benhelm}},
    \bibinfo {author} {\bibfnamefont {D.}~\bibnamefont {{Chek-al-kar}}}, \bibinfo
    {author} {\bibfnamefont {M.}~\bibnamefont {Chwalla}}, \bibinfo {author}
    {\bibfnamefont {T.}~\bibnamefont {K\"{o}rber}}, \bibinfo {author}
    {\bibfnamefont {U.~D.}\ \bibnamefont {Rapol}}, \bibinfo {author}
    {\bibfnamefont {M.}~\bibnamefont {Riebe}}, \bibinfo {author} {\bibfnamefont
    {P.~O.}\ \bibnamefont {Schmidt}}, \bibinfo {author} {\bibfnamefont
    {C.}~\bibnamefont {Becher}}, \bibinfo {author} {\bibfnamefont
    {O.}~\bibnamefont {G\"{u}hne}}, \bibinfo {author} {\bibfnamefont
    {W.}~\bibnamefont {D\"{u}r}},\ and\ \bibinfo {author} {\bibfnamefont
    {R.}~\bibnamefont {Blatt}},\ }\bibfield  {title} {\bibinfo {title} {Scalable
    multiparticle entanglement of trapped ions},\ }\href
    {https://doi.org/10.1038/nature04279} {\bibfield  {journal} {\bibinfo
    {journal} {Nature}\ }\textbf {\bibinfo {volume} {438}},\ \bibinfo {pages}
    {643} (\bibinfo {year} {2005})}\BibitemShut {NoStop}%
  \bibitem [{\citenamefont {Shang}\ \emph {et~al.}(2017)\citenamefont {Shang},
    \citenamefont {Zhang},\ and\ \citenamefont {Ng}}]{Shang.etal2017}%
    \BibitemOpen
    \bibfield  {author} {\bibinfo {author} {\bibfnamefont {J.}~\bibnamefont
    {Shang}}, \bibinfo {author} {\bibfnamefont {Z.}~\bibnamefont {Zhang}},\ and\
    \bibinfo {author} {\bibfnamefont {H.~K.}\ \bibnamefont {Ng}},\ }\bibfield
    {title} {\bibinfo {title} {Superfast maximum-likelihood reconstruction for
    quantum tomography},\ }\href {https://doi.org/10.1103/PhysRevA.95.062336}
    {\bibfield  {journal} {\bibinfo  {journal} {Phys. Rev. A}\ }\textbf {\bibinfo
    {volume} {95}},\ \bibinfo {pages} {062336} (\bibinfo {year}
    {2017})}\BibitemShut {NoStop}%
  \bibitem [{\citenamefont {Mayers}\ and\ \citenamefont
    {Yao}(2004)}]{Mayers.etal2004}%
    \BibitemOpen
    \bibfield  {author} {\bibinfo {author} {\bibfnamefont {D.}~\bibnamefont
    {Mayers}}\ and\ \bibinfo {author} {\bibfnamefont {A.}~\bibnamefont {Yao}},\
    }\bibfield  {title} {\bibinfo {title} {Self testing quantum apparatus},\
    }\href@noop {} {\bibfield  {journal} {\bibinfo  {journal} {Quantum Inf.
    Comput.}\ }\textbf {\bibinfo {volume} {4}},\ \bibinfo {pages} {273} (\bibinfo
    {year} {2004})}\BibitemShut {NoStop}%
  \bibitem [{\citenamefont {T\'oth}\ and\ \citenamefont
    {G\"uhne}(2005)}]{Toth.Guehne2005c}%
    \BibitemOpen
    \bibfield  {author} {\bibinfo {author} {\bibfnamefont {G.}~\bibnamefont
    {T\'oth}}\ and\ \bibinfo {author} {\bibfnamefont {O.}~\bibnamefont
    {G\"uhne}},\ }\bibfield  {title} {\bibinfo {title} {Detecting genuine
    multipartite entanglement with two local measurements},\ }\href
    {https://doi.org/10.1103/PhysRevLett.94.060501} {\bibfield  {journal}
    {\bibinfo  {journal} {Phys. Rev. Lett.}\ }\textbf {\bibinfo {volume} {94}},\
    \bibinfo {pages} {060501} (\bibinfo {year} {2005})}\BibitemShut {NoStop}%
  \bibitem [{\citenamefont {G{\"u}hne}\ and\ \citenamefont
    {T{\'o}th}(2009)}]{Guehne.Toth2009}%
    \BibitemOpen
    \bibfield  {author} {\bibinfo {author} {\bibfnamefont {O.}~\bibnamefont
    {G{\"u}hne}}\ and\ \bibinfo {author} {\bibfnamefont {G.}~\bibnamefont
    {T{\'o}th}},\ }\bibfield  {title} {\bibinfo {title} {Entanglement
    detection},\ }\href {https://doi.org/10.1016/j.physrep.2009.02.004}
    {\bibfield  {journal} {\bibinfo  {journal} {Phys. Rep.}\ }\textbf {\bibinfo
    {volume} {474}},\ \bibinfo {pages} {1} (\bibinfo {year} {2009})}\BibitemShut
    {NoStop}%
  \bibitem [{\citenamefont {Flammia}\ and\ \citenamefont
    {Liu}(2011)}]{Flammia.Liu2011}%
    \BibitemOpen
    \bibfield  {author} {\bibinfo {author} {\bibfnamefont {S.~T.}\ \bibnamefont
    {Flammia}}\ and\ \bibinfo {author} {\bibfnamefont {Y.-K.}\ \bibnamefont
    {Liu}},\ }\bibfield  {title} {\bibinfo {title} {Direct fidelity estimation
    from few {P}auli measurements},\ }\href
    {https://doi.org/10.1103/PhysRevLett.106.230501} {\bibfield  {journal}
    {\bibinfo  {journal} {Phys. Rev. Lett.}\ }\textbf {\bibinfo {volume} {106}},\
    \bibinfo {pages} {230501} (\bibinfo {year} {2011})}\BibitemShut {NoStop}%
  \bibitem [{\citenamefont {Dimi{\'c}}\ and\ \citenamefont
    {Daki{\'c}}(2018)}]{Dimic.Dakic2018}%
    \BibitemOpen
    \bibfield  {author} {\bibinfo {author} {\bibfnamefont {A.}~\bibnamefont
    {Dimi{\'c}}}\ and\ \bibinfo {author} {\bibfnamefont {B.}~\bibnamefont
    {Daki{\'c}}},\ }\bibfield  {title} {\bibinfo {title} {Single-copy
    entanglement detection},\ }\href {https://doi.org/10.1038/s41534-017-0055-x}
    {\bibfield  {journal} {\bibinfo  {journal} {npj Quantum Inf.}\ }\textbf
    {\bibinfo {volume} {4}},\ \bibinfo {pages} {11} (\bibinfo {year}
    {2018})}\BibitemShut {NoStop}%
  \bibitem [{\citenamefont {Pallister}\ \emph {et~al.}(2018)\citenamefont
    {Pallister}, \citenamefont {Linden},\ and\ \citenamefont
    {Montanaro}}]{Pallister.etal2018}%
    \BibitemOpen
    \bibfield  {author} {\bibinfo {author} {\bibfnamefont {S.}~\bibnamefont
    {Pallister}}, \bibinfo {author} {\bibfnamefont {N.}~\bibnamefont {Linden}},\
    and\ \bibinfo {author} {\bibfnamefont {A.}~\bibnamefont {Montanaro}},\
    }\bibfield  {title} {\bibinfo {title} {Optimal verification of entangled
    states with local measurements},\ }\href
    {https://doi.org/10.1103/PhysRevLett.120.170502} {\bibfield  {journal}
    {\bibinfo  {journal} {Phys. Rev. Lett.}\ }\textbf {\bibinfo {volume} {120}},\
    \bibinfo {pages} {170502} (\bibinfo {year} {2018})}\BibitemShut {NoStop}%
  \bibitem [{\citenamefont {Morimae}\ \emph {et~al.}(2017)\citenamefont
    {Morimae}, \citenamefont {Takeuchi},\ and\ \citenamefont
    {Hayashi}}]{Morimae.etal2017}%
    \BibitemOpen
    \bibfield  {author} {\bibinfo {author} {\bibfnamefont {T.}~\bibnamefont
    {Morimae}}, \bibinfo {author} {\bibfnamefont {Y.}~\bibnamefont {Takeuchi}},\
    and\ \bibinfo {author} {\bibfnamefont {M.}~\bibnamefont {Hayashi}},\
    }\bibfield  {title} {\bibinfo {title} {Verification of hypergraph states},\
    }\href {https://doi.org/10.1103/PhysRevA.96.062321} {\bibfield  {journal}
    {\bibinfo  {journal} {Phys. Rev. A}\ }\textbf {\bibinfo {volume} {96}},\
    \bibinfo {pages} {062321} (\bibinfo {year} {2017})}\BibitemShut {NoStop}%
  \bibitem [{\citenamefont {Takeuchi}\ and\ \citenamefont
    {Morimae}(2018)}]{Takeuchi.Morimae2018}%
    \BibitemOpen
    \bibfield  {author} {\bibinfo {author} {\bibfnamefont {Y.}~\bibnamefont
    {Takeuchi}}\ and\ \bibinfo {author} {\bibfnamefont {T.}~\bibnamefont
    {Morimae}},\ }\bibfield  {title} {\bibinfo {title} {Verification of
    many-qubit states},\ }\href {https://doi.org/10.1103/PhysRevX.8.021060}
    {\bibfield  {journal} {\bibinfo  {journal} {Phys. Rev. X}\ }\textbf {\bibinfo
    {volume} {8}},\ \bibinfo {pages} {021060} (\bibinfo {year}
    {2018})}\BibitemShut {NoStop}%
  \bibitem [{\citenamefont {Yu}\ \emph {et~al.}(2019)\citenamefont {Yu},
    \citenamefont {Shang},\ and\ \citenamefont {G\"uhne}}]{Yu.etal2019}%
    \BibitemOpen
    \bibfield  {author} {\bibinfo {author} {\bibfnamefont {X.-D.}\ \bibnamefont
    {Yu}}, \bibinfo {author} {\bibfnamefont {J.}~\bibnamefont {Shang}},\ and\
    \bibinfo {author} {\bibfnamefont {O.}~\bibnamefont {G\"uhne}},\ }\bibfield
    {title} {\bibinfo {title} {Optimal verification of general bipartite pure
    states},\ }\href {https://doi.org/10.1038/s41534-019-0226-z} {\bibfield
    {journal} {\bibinfo  {journal} {npj Quantum Inf.}\ }\textbf {\bibinfo
    {volume} {5}},\ \bibinfo {pages} {112} (\bibinfo {year} {2019})}\BibitemShut
    {NoStop}%
  \bibitem [{\citenamefont {Li}\ \emph {et~al.}(2019)\citenamefont {Li},
    \citenamefont {Han},\ and\ \citenamefont {Zhu}}]{Li.etal2019}%
    \BibitemOpen
    \bibfield  {author} {\bibinfo {author} {\bibfnamefont {Z.}~\bibnamefont
    {Li}}, \bibinfo {author} {\bibfnamefont {Y.-G.}\ \bibnamefont {Han}},\ and\
    \bibinfo {author} {\bibfnamefont {H.}~\bibnamefont {Zhu}},\ }\bibfield
    {title} {\bibinfo {title} {Efficient verification of bipartite pure states},\
    }\href {https://doi.org/10.1103/PhysRevA.100.032316} {\bibfield  {journal}
    {\bibinfo  {journal} {Phys. Rev. A}\ }\textbf {\bibinfo {volume} {100}},\
    \bibinfo {pages} {032316} (\bibinfo {year} {2019})}\BibitemShut {NoStop}%
  \bibitem [{\citenamefont {Wang}\ and\ \citenamefont
    {Hayashi}(2019)}]{Wang.Hayashi2019}%
    \BibitemOpen
    \bibfield  {author} {\bibinfo {author} {\bibfnamefont {K.}~\bibnamefont
    {Wang}}\ and\ \bibinfo {author} {\bibfnamefont {M.}~\bibnamefont {Hayashi}},\
    }\bibfield  {title} {\bibinfo {title} {Optimal verification of two-qubit pure
    states},\ }\href {https://doi.org/10.1103/PhysRevA.100.032315} {\bibfield
    {journal} {\bibinfo  {journal} {Phys. Rev. A}\ }\textbf {\bibinfo {volume}
    {100}},\ \bibinfo {pages} {032315} (\bibinfo {year} {2019})}\BibitemShut
    {NoStop}%
  \bibitem [{\citenamefont {{Zhu}}\ and\ \citenamefont
    {{Hayashi}}(2019)}]{Zhu.Hayashi2019a}%
    \BibitemOpen
    \bibfield  {author} {\bibinfo {author} {\bibfnamefont {H.}~\bibnamefont
    {{Zhu}}}\ and\ \bibinfo {author} {\bibfnamefont {M.}~\bibnamefont
    {{Hayashi}}},\ }\bibfield  {title} {\bibinfo {title} {Optimal verification
    and fidelity estimation of maximally entangled states},\ }\href
    {https://doi.org/10.1103/PhysRevA.99.052346} {\bibfield  {journal} {\bibinfo
    {journal} {Phys. Rev. A}\ }\textbf {\bibinfo {volume} {99}},\ \bibinfo
    {pages} {052346} (\bibinfo {year} {2019})}\BibitemShut {NoStop}%
  \bibitem [{\citenamefont {Zhu}\ and\ \citenamefont
    {Hayashi}(2019{\natexlab{a}})}]{Zhu.Hayashi2019b}%
    \BibitemOpen
    \bibfield  {author} {\bibinfo {author} {\bibfnamefont {H.}~\bibnamefont
    {Zhu}}\ and\ \bibinfo {author} {\bibfnamefont {M.}~\bibnamefont {Hayashi}},\
    }\bibfield  {title} {\bibinfo {title} {Efficient verification of hypergraph
    states},\ }\href {https://doi.org/10.1103/PhysRevApplied.12.054047}
    {\bibfield  {journal} {\bibinfo  {journal} {Phys. Rev. Appl.}\ }\textbf
    {\bibinfo {volume} {12}},\ \bibinfo {pages} {054047} (\bibinfo {year}
    {2019}{\natexlab{a}})}\BibitemShut {NoStop}%
  \bibitem [{\citenamefont {Zhu}\ and\ \citenamefont
    {Hayashi}(2019{\natexlab{b}})}]{Zhu.Hayashi2019c}%
    \BibitemOpen
    \bibfield  {author} {\bibinfo {author} {\bibfnamefont {H.}~\bibnamefont
    {Zhu}}\ and\ \bibinfo {author} {\bibfnamefont {M.}~\bibnamefont {Hayashi}},\
    }\bibfield  {title} {\bibinfo {title} {Efficient verification of pure quantum
    states in the adversarial scenario},\ }\href
    {https://doi.org/10.1103/PhysRevLett.123.260504} {\bibfield  {journal}
    {\bibinfo  {journal} {Phys. Rev. Lett.}\ }\textbf {\bibinfo {volume} {123}},\
    \bibinfo {pages} {260504} (\bibinfo {year} {2019}{\natexlab{b}})}\BibitemShut
    {NoStop}%
  \bibitem [{\citenamefont {Zhu}\ and\ \citenamefont
    {Hayashi}(2019{\natexlab{c}})}]{Zhu.Hayashi2019d}%
    \BibitemOpen
    \bibfield  {author} {\bibinfo {author} {\bibfnamefont {H.}~\bibnamefont
    {Zhu}}\ and\ \bibinfo {author} {\bibfnamefont {M.}~\bibnamefont {Hayashi}},\
    }\bibfield  {title} {\bibinfo {title} {General framework for verifying pure
    quantum states in the adversarial scenario},\ }\href
    {https://doi.org/10.1103/PhysRevA.100.062335} {\bibfield  {journal} {\bibinfo
     {journal} {Phys. Rev. A}\ }\textbf {\bibinfo {volume} {100}},\ \bibinfo
    {pages} {062335} (\bibinfo {year} {2019}{\natexlab{c}})}\BibitemShut
    {NoStop}%
  \bibitem [{\citenamefont {Liu}\ \emph {et~al.}(2019)\citenamefont {Liu},
    \citenamefont {Yu}, \citenamefont {Shang}, \citenamefont {Zhu},\ and\
    \citenamefont {Zhang}}]{Liu.etal2019b}%
    \BibitemOpen
    \bibfield  {author} {\bibinfo {author} {\bibfnamefont {Y.-C.}\ \bibnamefont
    {Liu}}, \bibinfo {author} {\bibfnamefont {X.-D.}\ \bibnamefont {Yu}},
    \bibinfo {author} {\bibfnamefont {J.}~\bibnamefont {Shang}}, \bibinfo
    {author} {\bibfnamefont {H.}~\bibnamefont {Zhu}},\ and\ \bibinfo {author}
    {\bibfnamefont {X.}~\bibnamefont {Zhang}},\ }\bibfield  {title} {\bibinfo
    {title} {Efficient verification of {D}icke states},\ }\href
    {https://doi.org/10.1103/PhysRevApplied.12.044020} {\bibfield  {journal}
    {\bibinfo  {journal} {Phys. Rev. Appl.}\ }\textbf {\bibinfo {volume} {12}},\
    \bibinfo {pages} {044020} (\bibinfo {year} {2019})}\BibitemShut {NoStop}%
  \bibitem [{\citenamefont {Li}\ \emph {et~al.}(2020)\citenamefont {Li},
    \citenamefont {Han},\ and\ \citenamefont {Zhu}}]{Li.etal2020b}%
    \BibitemOpen
    \bibfield  {author} {\bibinfo {author} {\bibfnamefont {Z.}~\bibnamefont
    {Li}}, \bibinfo {author} {\bibfnamefont {Y.-G.}\ \bibnamefont {Han}},\ and\
    \bibinfo {author} {\bibfnamefont {H.}~\bibnamefont {Zhu}},\ }\bibfield
    {title} {\bibinfo {title} {Optimal verification of
    {G}reenberger-{H}orne-{Z}eilinger states},\ }\href
    {https://doi.org/10.1103/PhysRevApplied.13.054002} {\bibfield  {journal}
    {\bibinfo  {journal} {Phys. Rev. Appl.}\ }\textbf {\bibinfo {volume} {13}},\
    \bibinfo {pages} {054002} (\bibinfo {year} {2020})}\BibitemShut {NoStop}%
  \bibitem [{\citenamefont {Dangniam}\ \emph {et~al.}(2020)\citenamefont
    {Dangniam}, \citenamefont {Han},\ and\ \citenamefont
    {Zhu}}]{Dangniam.etal2020}%
    \BibitemOpen
    \bibfield  {author} {\bibinfo {author} {\bibfnamefont {N.}~\bibnamefont
    {Dangniam}}, \bibinfo {author} {\bibfnamefont {Y.-G.}\ \bibnamefont {Han}},\
    and\ \bibinfo {author} {\bibfnamefont {H.}~\bibnamefont {Zhu}},\ }\bibfield
    {title} {\bibinfo {title} {Optimal verification of stabilizer states},\
    }\href {https://doi.org/10.1103/PhysRevResearch.2.043323} {\bibfield
    {journal} {\bibinfo  {journal} {Phys. Rev. Research}\ }\textbf {\bibinfo
    {volume} {2}},\ \bibinfo {pages} {043323} (\bibinfo {year}
    {2020})}\BibitemShut {NoStop}%
  \bibitem [{\citenamefont {Zhang}\ \emph
    {et~al.}(2020{\natexlab{a}})\citenamefont {Zhang}, \citenamefont {Zhang},
    \citenamefont {Chen}, \citenamefont {Peng}, \citenamefont {Xu}, \citenamefont
    {Yin}, \citenamefont {Yu}, \citenamefont {Ye}, \citenamefont {Han},
    \citenamefont {Xu}, \citenamefont {Chen}, \citenamefont {Li},\ and\
    \citenamefont {Guo}}]{Zhang.etal2020a}%
    \BibitemOpen
    \bibfield  {author} {\bibinfo {author} {\bibfnamefont {W.-H.}\ \bibnamefont
    {Zhang}}, \bibinfo {author} {\bibfnamefont {C.}~\bibnamefont {Zhang}},
    \bibinfo {author} {\bibfnamefont {Z.}~\bibnamefont {Chen}}, \bibinfo {author}
    {\bibfnamefont {X.-X.}\ \bibnamefont {Peng}}, \bibinfo {author}
    {\bibfnamefont {X.-Y.}\ \bibnamefont {Xu}}, \bibinfo {author} {\bibfnamefont
    {P.}~\bibnamefont {Yin}}, \bibinfo {author} {\bibfnamefont {S.}~\bibnamefont
    {Yu}}, \bibinfo {author} {\bibfnamefont {X.-J.}\ \bibnamefont {Ye}}, \bibinfo
    {author} {\bibfnamefont {Y.-J.}\ \bibnamefont {Han}}, \bibinfo {author}
    {\bibfnamefont {J.-S.}\ \bibnamefont {Xu}}, \bibinfo {author} {\bibfnamefont
    {G.}~\bibnamefont {Chen}}, \bibinfo {author} {\bibfnamefont {C.-F.}\
    \bibnamefont {Li}},\ and\ \bibinfo {author} {\bibfnamefont {G.-C.}\
    \bibnamefont {Guo}},\ }\bibfield  {title} {\bibinfo {title} {Experimental
    optimal verification of entangled states using local measurements},\ }\href
    {https://doi.org/10.1103/PhysRevLett.125.030506} {\bibfield  {journal}
    {\bibinfo  {journal} {Phys. Rev. Lett.}\ }\textbf {\bibinfo {volume} {125}},\
    \bibinfo {pages} {030506} (\bibinfo {year} {2020}{\natexlab{a}})}\BibitemShut
    {NoStop}%
  \bibitem [{\citenamefont {Jiang}\ \emph {et~al.}(2020)\citenamefont {Jiang},
    \citenamefont {Wang}, \citenamefont {Qian}, \citenamefont {Chen},
    \citenamefont {Chen}, \citenamefont {Lu}, \citenamefont {Xia}, \citenamefont
    {Song}, \citenamefont {Zhu},\ and\ \citenamefont {Ma}}]{Jiang.etal2020}%
    \BibitemOpen
    \bibfield  {author} {\bibinfo {author} {\bibfnamefont {X.}~\bibnamefont
    {Jiang}}, \bibinfo {author} {\bibfnamefont {K.}~\bibnamefont {Wang}},
    \bibinfo {author} {\bibfnamefont {K.}~\bibnamefont {Qian}}, \bibinfo {author}
    {\bibfnamefont {Z.}~\bibnamefont {Chen}}, \bibinfo {author} {\bibfnamefont
    {Z.}~\bibnamefont {Chen}}, \bibinfo {author} {\bibfnamefont {L.}~\bibnamefont
    {Lu}}, \bibinfo {author} {\bibfnamefont {L.}~\bibnamefont {Xia}}, \bibinfo
    {author} {\bibfnamefont {F.}~\bibnamefont {Song}}, \bibinfo {author}
    {\bibfnamefont {S.}~\bibnamefont {Zhu}},\ and\ \bibinfo {author}
    {\bibfnamefont {X.}~\bibnamefont {Ma}},\ }\bibfield  {title} {\bibinfo
    {title} {Towards the standardization of quantum state verification using
    optimal strategies},\ }\href {https://doi.org/10.1038/s41534-020-00317-7}
    {\bibfield  {journal} {\bibinfo  {journal} {npj Quantum Inf.}\ }\textbf
    {\bibinfo {volume} {6}},\ \bibinfo {pages} {90} (\bibinfo {year}
    {2020})}\BibitemShut {NoStop}%
  \bibitem [{\citenamefont {Zhang}\ \emph
    {et~al.}(2020{\natexlab{b}})\citenamefont {Zhang}, \citenamefont {Liu},
    \citenamefont {Yin}, \citenamefont {Peng}, \citenamefont {Li}, \citenamefont
    {Xu}, \citenamefont {Yu}, \citenamefont {Hou}, \citenamefont {Han},
    \citenamefont {Xu}, \citenamefont {Zhou}, \citenamefont {Chen}, \citenamefont
    {Li},\ and\ \citenamefont {Guo}}]{Zhang.etal2020b}%
    \BibitemOpen
    \bibfield  {author} {\bibinfo {author} {\bibfnamefont {W.-H.}\ \bibnamefont
    {Zhang}}, \bibinfo {author} {\bibfnamefont {X.}~\bibnamefont {Liu}}, \bibinfo
    {author} {\bibfnamefont {P.}~\bibnamefont {Yin}}, \bibinfo {author}
    {\bibfnamefont {X.-X.}\ \bibnamefont {Peng}}, \bibinfo {author}
    {\bibfnamefont {G.-C.}\ \bibnamefont {Li}}, \bibinfo {author} {\bibfnamefont
    {X.-Y.}\ \bibnamefont {Xu}}, \bibinfo {author} {\bibfnamefont
    {S.}~\bibnamefont {Yu}}, \bibinfo {author} {\bibfnamefont {Z.-B.}\
    \bibnamefont {Hou}}, \bibinfo {author} {\bibfnamefont {Y.-J.}\ \bibnamefont
    {Han}}, \bibinfo {author} {\bibfnamefont {J.-S.}\ \bibnamefont {Xu}},
    \bibinfo {author} {\bibfnamefont {Z.-Q.}\ \bibnamefont {Zhou}}, \bibinfo
    {author} {\bibfnamefont {G.}~\bibnamefont {Chen}}, \bibinfo {author}
    {\bibfnamefont {C.-F.}\ \bibnamefont {Li}},\ and\ \bibinfo {author}
    {\bibfnamefont {G.-C.}\ \bibnamefont {Guo}},\ }\bibfield  {title} {\bibinfo
    {title} {Classical communication enhanced quantum state verification},\
    }\href {https://doi.org/10.1038/s41534-020-00328-4} {\bibfield  {journal}
    {\bibinfo  {journal} {npj Quantum Inf.}\ }\textbf {\bibinfo {volume} {6}},\
    \bibinfo {pages} {103} (\bibinfo {year} {2020}{\natexlab{b}})}\BibitemShut
    {NoStop}%
  \bibitem [{\citenamefont {Li}\ \emph {et~al.}(2021{\natexlab{a}})\citenamefont
    {Li}, \citenamefont {Han}, \citenamefont {Sun}, \citenamefont {Shang},\ and\
    \citenamefont {Zhu}}]{Li.etal2020a}%
    \BibitemOpen
    \bibfield  {author} {\bibinfo {author} {\bibfnamefont {Z.}~\bibnamefont
    {Li}}, \bibinfo {author} {\bibfnamefont {Y.-G.}\ \bibnamefont {Han}},
    \bibinfo {author} {\bibfnamefont {H.-F.}\ \bibnamefont {Sun}}, \bibinfo
    {author} {\bibfnamefont {J.}~\bibnamefont {Shang}},\ and\ \bibinfo {author}
    {\bibfnamefont {H.}~\bibnamefont {Zhu}},\ }\bibfield  {title} {\bibinfo
    {title} {Verification of phased {D}icke states},\ }\href
    {https://doi.org/10.1103/PhysRevA.103.022601} {\bibfield  {journal} {\bibinfo
     {journal} {Phys. Rev. A}\ }\textbf {\bibinfo {volume} {103}},\ \bibinfo
    {pages} {022601} (\bibinfo {year} {2021}{\natexlab{a}})}\BibitemShut
    {NoStop}%
  \bibitem [{\citenamefont {Liu}\ \emph {et~al.}(2021{\natexlab{a}})\citenamefont
    {Liu}, \citenamefont {Shang}, \citenamefont {Han},\ and\ \citenamefont
    {Zhang}}]{Liu.etal2020b}%
    \BibitemOpen
    \bibfield  {author} {\bibinfo {author} {\bibfnamefont {Y.-C.}\ \bibnamefont
    {Liu}}, \bibinfo {author} {\bibfnamefont {J.}~\bibnamefont {Shang}}, \bibinfo
    {author} {\bibfnamefont {R.}~\bibnamefont {Han}},\ and\ \bibinfo {author}
    {\bibfnamefont {X.}~\bibnamefont {Zhang}},\ }\bibfield  {title} {\bibinfo
    {title} {Universally optimal verification of entangled states with
    nondemolition measurements},\ }\href
    {https://doi.org/10.1103/PhysRevLett.126.090504} {\bibfield  {journal}
    {\bibinfo  {journal} {Phys. Rev. Lett.}\ }\textbf {\bibinfo {volume} {126}},\
    \bibinfo {pages} {090504} (\bibinfo {year} {2021}{\natexlab{a}})}\BibitemShut
    {NoStop}%
  \bibitem [{\citenamefont {Liu}\ \emph {et~al.}(2021{\natexlab{b}})\citenamefont
    {Liu}, \citenamefont {Shang},\ and\ \citenamefont {Zhang}}]{Liu.etal2021}%
    \BibitemOpen
    \bibfield  {author} {\bibinfo {author} {\bibfnamefont {Y.-C.}\ \bibnamefont
    {Liu}}, \bibinfo {author} {\bibfnamefont {J.}~\bibnamefont {Shang}},\ and\
    \bibinfo {author} {\bibfnamefont {X.}~\bibnamefont {Zhang}},\ }\bibfield
    {title} {\bibinfo {title} {Efficient verification of entangled
    continuous-variable quantum states with local measurements},\ }\href
    {https://doi.org/10.1103/PhysRevResearch.3.L042004} {\bibfield  {journal}
    {\bibinfo  {journal} {Phys. Rev. Research}\ }\textbf {\bibinfo {volume}
    {3}},\ \bibinfo {pages} {L042004} (\bibinfo {year}
    {2021}{\natexlab{b}})}\BibitemShut {NoStop}%
  \bibitem [{\citenamefont {Han}\ \emph {et~al.}(2021)\citenamefont {Han},
    \citenamefont {Li}, \citenamefont {Wang},\ and\ \citenamefont
    {Zhu}}]{Han.etal2021}%
    \BibitemOpen
    \bibfield  {author} {\bibinfo {author} {\bibfnamefont {Y.-G.}\ \bibnamefont
    {Han}}, \bibinfo {author} {\bibfnamefont {Z.}~\bibnamefont {Li}}, \bibinfo
    {author} {\bibfnamefont {Y.}~\bibnamefont {Wang}},\ and\ \bibinfo {author}
    {\bibfnamefont {H.}~\bibnamefont {Zhu}},\ }\bibfield  {title} {\bibinfo
    {title} {Optimal verification of the {B}ell state and
    {G}reenberger--{H}orne--{Z}eilinger states in untrusted quantum networks},\
    }\href {https://doi.org/https://doi.org/10.1038/s41534-021-00499-8}
    {\bibfield  {journal} {\bibinfo  {journal} {npj Quantum Inf.}\ }\textbf
    {\bibinfo {volume} {7}},\ \bibinfo {pages} {164} (\bibinfo {year}
    {2021})}\BibitemShut {NoStop}%
  \bibitem [{\citenamefont {Zhu}\ \emph {et~al.}(2022)\citenamefont {Zhu},
    \citenamefont {Li},\ and\ \citenamefont {Chen}}]{Zhu.etal2022}%
    \BibitemOpen
    \bibfield  {author} {\bibinfo {author} {\bibfnamefont {H.}~\bibnamefont
    {Zhu}}, \bibinfo {author} {\bibfnamefont {Y.}~\bibnamefont {Li}},\ and\
    \bibinfo {author} {\bibfnamefont {T.}~\bibnamefont {Chen}},\ }\bibfield
    {title} {\bibinfo {title} {Efficient verification of ground states of
    frustration-free hamiltonians},\ }\href@noop {} {\bibfield  {journal}
    {\bibinfo  {journal} {arXiv preprint arXiv:2206.15292}\ } (\bibinfo {year}
    {2022})}\BibitemShut {NoStop}%
  \bibitem [{\citenamefont {Chen}\ \emph {et~al.}(2023)\citenamefont {Chen},
    \citenamefont {Li},\ and\ \citenamefont {Zhu}}]{ChenHuang.etal2023}%
    \BibitemOpen
    \bibfield  {author} {\bibinfo {author} {\bibfnamefont {T.}~\bibnamefont
    {Chen}}, \bibinfo {author} {\bibfnamefont {Y.}~\bibnamefont {Li}},\ and\
    \bibinfo {author} {\bibfnamefont {H.}~\bibnamefont {Zhu}},\ }\bibfield
    {title} {\bibinfo {title} {Efficient verification of
    {A}ffleck-{K}ennedy-{L}ieb-{T}asaki states},\ }\href
    {https://doi.org/10.1103/PhysRevA.107.022616} {\bibfield  {journal} {\bibinfo
     {journal} {Phys. Rev. A}\ }\textbf {\bibinfo {volume} {107}},\ \bibinfo
    {pages} {022616} (\bibinfo {year} {2023})}\BibitemShut {NoStop}%
  \bibitem [{\citenamefont {Li}\ \emph {et~al.}(2021{\natexlab{b}})\citenamefont
    {Li}, \citenamefont {Zhang}, \citenamefont {Li},\ and\ \citenamefont
    {Zhu}}]{Li.Y.etal2021}%
    \BibitemOpen
    \bibfield  {author} {\bibinfo {author} {\bibfnamefont {Y.}~\bibnamefont
    {Li}}, \bibinfo {author} {\bibfnamefont {H.}~\bibnamefont {Zhang}}, \bibinfo
    {author} {\bibfnamefont {Z.}~\bibnamefont {Li}},\ and\ \bibinfo {author}
    {\bibfnamefont {H.}~\bibnamefont {Zhu}},\ }\bibfield  {title} {\bibinfo
    {title} {Minimum number of experimental settings required to verify bipartite
    pure states and unitaries},\ }\href
    {https://doi.org/10.1103/PhysRevA.104.062439} {\bibfield  {journal} {\bibinfo
     {journal} {Phys. Rev. A}\ }\textbf {\bibinfo {volume} {104}},\ \bibinfo
    {pages} {062439} (\bibinfo {year} {2021}{\natexlab{b}})}\BibitemShut
    {NoStop}%
  \bibitem [{\citenamefont {Liu}\ \emph {et~al.}(2020)\citenamefont {Liu},
    \citenamefont {Shang}, \citenamefont {Yu},\ and\ \citenamefont
    {Zhang}}]{Liu.etal2020a}%
    \BibitemOpen
    \bibfield  {author} {\bibinfo {author} {\bibfnamefont {Y.-C.}\ \bibnamefont
    {Liu}}, \bibinfo {author} {\bibfnamefont {J.}~\bibnamefont {Shang}}, \bibinfo
    {author} {\bibfnamefont {X.-D.}\ \bibnamefont {Yu}},\ and\ \bibinfo {author}
    {\bibfnamefont {X.}~\bibnamefont {Zhang}},\ }\bibfield  {title} {\bibinfo
    {title} {Efficient verification of quantum processes},\ }\href
    {https://doi.org/10.1103/PhysRevA.101.042315} {\bibfield  {journal} {\bibinfo
     {journal} {Phys. Rev. A}\ }\textbf {\bibinfo {volume} {101}},\ \bibinfo
    {pages} {042315} (\bibinfo {year} {2020})}\BibitemShut {NoStop}%
  \bibitem [{\citenamefont {Zhu}\ and\ \citenamefont
    {Zhang}(2020)}]{Zhu.Zhang2020}%
    \BibitemOpen
    \bibfield  {author} {\bibinfo {author} {\bibfnamefont {H.}~\bibnamefont
    {Zhu}}\ and\ \bibinfo {author} {\bibfnamefont {H.}~\bibnamefont {Zhang}},\
    }\bibfield  {title} {\bibinfo {title} {Efficient verification of quantum
    gates with local operations},\ }\href
    {https://doi.org/10.1103/PhysRevA.101.042316} {\bibfield  {journal} {\bibinfo
     {journal} {Phys. Rev. A}\ }\textbf {\bibinfo {volume} {101}},\ \bibinfo
    {pages} {042316} (\bibinfo {year} {2020})}\BibitemShut {NoStop}%
  \bibitem [{\citenamefont {Zeng}\ \emph {et~al.}(2020)\citenamefont {Zeng},
    \citenamefont {Zhou},\ and\ \citenamefont {Liu}}]{Zeng.etal2020}%
    \BibitemOpen
    \bibfield  {author} {\bibinfo {author} {\bibfnamefont {P.}~\bibnamefont
    {Zeng}}, \bibinfo {author} {\bibfnamefont {Y.}~\bibnamefont {Zhou}},\ and\
    \bibinfo {author} {\bibfnamefont {Z.}~\bibnamefont {Liu}},\ }\bibfield
    {title} {\bibinfo {title} {Quantum gate verification and its application in
    property testing},\ }\href {https://doi.org/10.1103/PhysRevResearch.2.023306}
    {\bibfield  {journal} {\bibinfo  {journal} {Phys. Rev. Research}\ }\textbf
    {\bibinfo {volume} {2}},\ \bibinfo {pages} {023306} (\bibinfo {year}
    {2020})}\BibitemShut {NoStop}%
  \bibitem [{\citenamefont {Zhang}\ \emph {et~al.}(2022)\citenamefont {Zhang},
    \citenamefont {Hou}, \citenamefont {Tang}, \citenamefont {Shang},
    \citenamefont {Zhu}, \citenamefont {Xiang}, \citenamefont {Li},\ and\
    \citenamefont {Guo}}]{Zhang.etal2022}%
    \BibitemOpen
    \bibfield  {author} {\bibinfo {author} {\bibfnamefont {R.-Q.}\ \bibnamefont
    {Zhang}}, \bibinfo {author} {\bibfnamefont {Z.}~\bibnamefont {Hou}}, \bibinfo
    {author} {\bibfnamefont {J.-F.}\ \bibnamefont {Tang}}, \bibinfo {author}
    {\bibfnamefont {J.}~\bibnamefont {Shang}}, \bibinfo {author} {\bibfnamefont
    {H.}~\bibnamefont {Zhu}}, \bibinfo {author} {\bibfnamefont {G.-Y.}\
    \bibnamefont {Xiang}}, \bibinfo {author} {\bibfnamefont {C.-F.}\ \bibnamefont
    {Li}},\ and\ \bibinfo {author} {\bibfnamefont {G.-C.}\ \bibnamefont {Guo}},\
    }\bibfield  {title} {\bibinfo {title} {Efficient experimental verification of
    quantum gates with local operations},\ }\href
    {https://doi.org/10.1103/PhysRevLett.128.020502} {\bibfield  {journal}
    {\bibinfo  {journal} {Phys. Rev. Lett.}\ }\textbf {\bibinfo {volume} {128}},\
    \bibinfo {pages} {020502} (\bibinfo {year} {2022})}\BibitemShut {NoStop}%
  \bibitem [{\citenamefont {Yu}\ \emph {et~al.}(2022)\citenamefont {Yu},
    \citenamefont {Shang},\ and\ \citenamefont {Gühne}}]{Yu.etal2022}%
    \BibitemOpen
    \bibfield  {author} {\bibinfo {author} {\bibfnamefont {X.-D.}\ \bibnamefont
    {Yu}}, \bibinfo {author} {\bibfnamefont {J.}~\bibnamefont {Shang}},\ and\
    \bibinfo {author} {\bibfnamefont {O.}~\bibnamefont {Gühne}},\ }\bibfield
    {title} {\bibinfo {title} {Statistical methods for quantum state verification
    and fidelity estimation},\ }\href
    {https://doi.org/https://doi.org/10.1002/qute.202100126} {\bibfield
    {journal} {\bibinfo  {journal} {Adv. Quantum Technol.}\ }\textbf {\bibinfo
    {volume} {5}},\ \bibinfo {pages} {2100126} (\bibinfo {year}
    {2022})}\BibitemShut {NoStop}%
  \bibitem [{\citenamefont {Bourennane}\ \emph {et~al.}(2004)\citenamefont
    {Bourennane}, \citenamefont {Eibl}, \citenamefont {Kurtsiefer}, \citenamefont
    {Gaertner}, \citenamefont {Weinfurter}, \citenamefont {G\"uhne},
    \citenamefont {Hyllus}, \citenamefont {Bru\ss{}}, \citenamefont
    {Lewenstein},\ and\ \citenamefont {Sanpera}}]{Bourennane.etal2004}%
    \BibitemOpen
    \bibfield  {author} {\bibinfo {author} {\bibfnamefont {M.}~\bibnamefont
    {Bourennane}}, \bibinfo {author} {\bibfnamefont {M.}~\bibnamefont {Eibl}},
    \bibinfo {author} {\bibfnamefont {C.}~\bibnamefont {Kurtsiefer}}, \bibinfo
    {author} {\bibfnamefont {S.}~\bibnamefont {Gaertner}}, \bibinfo {author}
    {\bibfnamefont {H.}~\bibnamefont {Weinfurter}}, \bibinfo {author}
    {\bibfnamefont {O.}~\bibnamefont {G\"uhne}}, \bibinfo {author} {\bibfnamefont
    {P.}~\bibnamefont {Hyllus}}, \bibinfo {author} {\bibfnamefont
    {D.}~\bibnamefont {Bru\ss{}}}, \bibinfo {author} {\bibfnamefont
    {M.}~\bibnamefont {Lewenstein}},\ and\ \bibinfo {author} {\bibfnamefont
    {A.}~\bibnamefont {Sanpera}},\ }\bibfield  {title} {\bibinfo {title}
    {Experimental detection of multipartite entanglement using witness
    operators},\ }\href {https://doi.org/10.1103/PhysRevLett.92.087902}
    {\bibfield  {journal} {\bibinfo  {journal} {Phys. Rev. Lett.}\ }\textbf
    {\bibinfo {volume} {92}},\ \bibinfo {pages} {087902} (\bibinfo {year}
    {2004})}\BibitemShut {NoStop}%
  \bibitem [{\citenamefont {Brand\~ao}(2005)}]{Brandao.etal2005}%
    \BibitemOpen
    \bibfield  {author} {\bibinfo {author} {\bibfnamefont {F.~G. S.~L.}\
    \bibnamefont {Brand\~ao}},\ }\bibfield  {title} {\bibinfo {title}
    {Quantifying entanglement with witness operators},\ }\href
    {https://doi.org/10.1103/PhysRevA.72.022310} {\bibfield  {journal} {\bibinfo
    {journal} {Phys. Rev. A}\ }\textbf {\bibinfo {volume} {72}},\ \bibinfo
    {pages} {022310} (\bibinfo {year} {2005})}\BibitemShut {NoStop}%
  \end{thebibliography}
\end{document}